\documentclass[11pt, reqno]{amsart}
\usepackage{amssymb}
\usepackage{eucal}
\usepackage{amsmath}
\usepackage{amscd}
\usepackage[dvips]{color}
\usepackage{multicol}
\usepackage[all]{xy}           
\usepackage{graphicx}
\usepackage{color}
\usepackage{colordvi}
\usepackage{xspace}
\usepackage{tikz}
\usepackage[normalem]{ulem}

\usepackage[colorlinks,final,backref=page,hyperindex,hypertex]{hyperref}
\usepackage[active]{srcltx} 

\begin{document}  
\newcommand {\emptycomment}[1]{} 

\baselineskip=14pt
\newcommand{\nc}{\newcommand}
\newcommand{\delete}[1]{}
\nc{\mfootnote}[1]{\footnote{#1}} 
\nc{\todo}[1]{\tred{To do:} #1}

\nc{\mlabel}[1]{\label{#1}}  
\nc{\mcite}[1]{\cite{#1}}  
\nc{\mref}[1]{\ref{#1}}  
\nc{\mbibitem}[1]{\bibitem{#1}} 

\delete{
\nc{\mlabel}[1]{\label{#1}  
{\hfill \hspace{1cm}{\bf{{\ }\hfill(#1)}}}}
\nc{\mcite}[1]{\cite{#1}{{\bf{{\ }(#1)}}}}  
\nc{\mref}[1]{\ref{#1}{{\bf{{\ }(#1)}}}}  
\nc{\mbibitem}[1]{\bibitem[\bf #1]{#1}} 
}

\newtheorem{thm}{Theorem}[section]
\newtheorem{lem}[thm]{Lemma}
\newtheorem{cor}[thm]{Corollary}
\newtheorem{pro}[thm]{Proposition}
\newtheorem{ex}[thm]{Example}
\newtheorem{rmk}[thm]{Remark}
\newtheorem{defi}[thm]{Definition}
\newtheorem{pdef}[thm]{Proposition-Definition}
\newtheorem{condition}[thm]{Condition}

\renewcommand{\labelenumi}{{\rm(\alph{enumi})}}
\renewcommand{\theenumi}{\alph{enumi}}

\nc{\tred}[1]{\textcolor{red}{#1}}
\nc{\tblue}[1]{\textcolor{blue}{#1}}
\nc{\tgreen}[1]{\textcolor{green}{#1}}
\nc{\tpurple}[1]{\textcolor{purple}{#1}}
\nc{\btred}[1]{\textcolor{red}{\bf #1}}
\nc{\btblue}[1]{\textcolor{blue}{\bf #1}}
\nc{\btgreen}[1]{\textcolor{green}{\bf #1}}
\nc{\btpurple}[1]{\textcolor{purple}{\bf #1}}

\nc{\cm}[1]{\textcolor{red}{Chengming:#1}}
\nc{\li}[1]{\textcolor{blue}{Li: #1}}
\nc{\lit}[2]{\textcolor{blue}{#1}{ \textcolor{purple}{(\sout{#2})}}}
\nc{\yh}[1]{\textcolor{green}{Yunhe: #1}}


\nc{\twovec}[2]{\left(\begin{array}{c} #1 \\ #2\end{array} \right )}
\nc{\threevec}[3]{\left(\begin{array}{c} #1 \\ #2 \\ #3 \end{array}\right )}
\nc{\twomatrix}[4]{\left(\begin{array}{cc} #1 & #2\\ #3 & #4 \end{array} \right)}
\nc{\threematrix}[9]{{\left(\begin{matrix} #1 & #2 & #3\\ #4 & #5 & #6 \\ #7 & #8 & #9 \end{matrix} \right)}}
\nc{\twodet}[4]{\left|\begin{array}{cc} #1 & #2\\ #3 & #4 \end{array} \right|}

\nc{\rk}{\mathrm{r}}
\newcommand{\g}{\mathfrak g}
\newcommand{\h}{\mathfrak h}
\newcommand{\pf}{\noindent{$Proof$.}\ }
\newcommand{\frkg}{\mathfrak g}
\newcommand{\frkh}{\mathfrak h}
\newcommand{\Id}{\rm{Id}}
\newcommand{\gl}{\mathfrak {gl}}
\newcommand{\ad}{\mathrm{ad}}
\newcommand{\add}{\frka\frkd}
\newcommand{\frka}{\mathfrak a}
\newcommand{\frkb}{\mathfrak b}
\newcommand{\frkc}{\mathfrak c}
\newcommand{\frkd}{\mathfrak d}
\newcommand {\comment}[1]{{\marginpar{*}\scriptsize\textbf{Comments:} #1}}

\nc{\gensp}{V} 
\nc{\relsp}{\Lambda} 
\nc{\leafsp}{X}    
\nc{\treesp}{\overline{\calt}} 

\nc{\vin}{{\mathrm Vin}}    
\nc{\lin}{{\mathrm Lin}}    

\nc{\gop}{{\,\omega\,}}     
\nc{\gopb}{{\,\nu\,}}
\nc{\svec}[2]{{\tiny\left(\begin{matrix}#1\\
#2\end{matrix}\right)\,}}  
\nc{\ssvec}[2]{{\tiny\left(\begin{matrix}#1\\
#2\end{matrix}\right)\,}} 

\nc{\typeI}{local cocycle $3$-Lie bialgebra\xspace}
\nc{\typeIs}{local cocycle $3$-Lie bialgebras\xspace}
\nc{\typeII}{double construction $3$-Lie bialgebra\xspace}
\nc{\typeIIs}{double construction $3$-Lie bialgebras\xspace}

\nc{\bia}{{$\mathcal{P}$-bimodule ${\bf k}$-algebra}\xspace}
\nc{\bias}{{$\mathcal{P}$-bimodule ${\bf k}$-algebras}\xspace}

\nc{\rmi}{{\mathrm{I}}}
\nc{\rmii}{{\mathrm{II}}}
\nc{\rmiii}{{\mathrm{III}}}
\nc{\pr}{{\mathrm{pr}}}
\newcommand{\huaA}{\mathcal{A}}

\nc{\tcybe}{3-Lie CYBE}

\nc{\pll}{\beta}
\nc{\plc}{\epsilon}

\nc{\ass}{{\mathit{Ass}}}
\nc{\lie}{{\mathit{Lie}}}
\nc{\comm}{{\mathit{Comm}}}
\nc{\dend}{{\mathit{Dend}}}
\nc{\zinb}{{\mathit{Zinb}}}
\nc{\tdend}{{\mathit{TDend}}}
\nc{\prelie}{{\mathit{preLie}}}
\nc{\postlie}{{\mathit{PostLie}}}
\nc{\quado}{{\mathit{Quad}}}
\nc{\octo}{{\mathit{Octo}}}
\nc{\ldend}{{\mathit{ldend}}}
\nc{\lquad}{{\mathit{LQuad}}}

 \nc{\adec}{\check{;}} \nc{\aop}{\alpha}
\nc{\dftimes}{\widetilde{\otimes}} \nc{\dfl}{\succ} \nc{\dfr}{\prec}
\nc{\dfc}{\circ} \nc{\dfb}{\bullet} \nc{\dft}{\star}
\nc{\dfcf}{{\mathbf k}} \nc{\apr}{\ast} \nc{\spr}{\cdot}
\nc{\twopr}{\circ} \nc{\tspr}{\star} \nc{\sempr}{\ast}
\nc{\disp}[1]{\displaystyle{#1}}
\nc{\bin}[2]{ (_{\stackrel{\scs{#1}}{\scs{#2}}})}  
\nc{\binc}[2]{ \left (\!\! \begin{array}{c} \scs{#1}\\
    \scs{#2} \end{array}\!\! \right )}  
\nc{\bincc}[2]{  \left ( {\scs{#1} \atop
    \vspace{-.5cm}\scs{#2}} \right )}  
\nc{\sarray}[2]{\begin{array}{c}#1 \vspace{.1cm}\\ \hline
    \vspace{-.35cm} \\ #2 \end{array}}
\nc{\bs}{\bar{S}} \nc{\dcup}{\stackrel{\bullet}{\cup}}
\nc{\dbigcup}{\stackrel{\bullet}{\bigcup}} \nc{\etree}{\big |}
\nc{\la}{\longrightarrow} \nc{\fe}{\'{e}} \nc{\rar}{\rightarrow}
\nc{\dar}{\downarrow} \nc{\dap}[1]{\downarrow
\rlap{$\scriptstyle{#1}$}} \nc{\uap}[1]{\uparrow
\rlap{$\scriptstyle{#1}$}} \nc{\defeq}{\stackrel{\rm def}{=}}
\nc{\dis}[1]{\displaystyle{#1}} \nc{\dotcup}{\,
\displaystyle{\bigcup^\bullet}\ } \nc{\sdotcup}{\tiny{
\displaystyle{\bigcup^\bullet}\ }} \nc{\hcm}{\ \hat{,}\ }
\nc{\hcirc}{\hat{\circ}} \nc{\hts}{\hat{\shpr}}
\nc{\lts}{\stackrel{\leftarrow}{\shpr}}
\nc{\rts}{\stackrel{\rightarrow}{\shpr}} \nc{\lleft}{[}
\nc{\lright}{]} \nc{\uni}[1]{\tilde{#1}} \nc{\wor}[1]{\check{#1}}
\nc{\free}[1]{\bar{#1}} \nc{\den}[1]{\check{#1}} \nc{\lrpa}{\wr}
\nc{\curlyl}{\left \{ \begin{array}{c} {} \\ {} \end{array}
    \right .  \!\!\!\!\!\!\!}
\nc{\curlyr}{ \!\!\!\!\!\!\!
    \left . \begin{array}{c} {} \\ {} \end{array}
    \right \} }
\nc{\leaf}{\ell}       
\nc{\longmid}{\left | \begin{array}{c} {} \\ {} \end{array}
    \right . \!\!\!\!\!\!\!}
\nc{\ot}{\otimes} \nc{\sot}{{\scriptstyle{\ot}}}
\nc{\otm}{\overline{\ot}}
\nc{\ora}[1]{\stackrel{#1}{\rar}}
\nc{\ola}[1]{\stackrel{#1}{\la}}
\nc{\pltree}{\calt^\pl}
\nc{\epltree}{\calt^{\pl,\NC}}
\nc{\rbpltree}{\calt^r}
\nc{\scs}[1]{\scriptstyle{#1}} \nc{\mrm}[1]{{\rm #1}}
\nc{\dirlim}{\displaystyle{\lim_{\longrightarrow}}\,}
\nc{\invlim}{\displaystyle{\lim_{\longleftarrow}}\,}
\nc{\mvp}{\vspace{0.5cm}} \nc{\svp}{\vspace{2cm}}
\nc{\vp}{\vspace{8cm}} \nc{\proofbegin}{\noindent{\bf Proof: }}
\nc{\proofend}{$\blacksquare$ \vspace{0.5cm}}
\nc{\freerbpl}{{F^{\mathrm RBPL}}}
\nc{\sha}{{\mbox{\cyr X}}}  
\nc{\ncsha}{{\mbox{\cyr X}^{\mathrm NC}}} \nc{\ncshao}{{\mbox{\cyr
X}^{\mathrm NC,\,0}}}
\nc{\shpr}{\diamond}    
\nc{\shprm}{\overline{\diamond}}    
\nc{\shpro}{\diamond^0}    
\nc{\shprr}{\diamond^r}     
\nc{\shpra}{\overline{\diamond}^r}
\nc{\shpru}{\check{\diamond}} \nc{\catpr}{\diamond_l}
\nc{\rcatpr}{\diamond_r} \nc{\lapr}{\diamond_a}
\nc{\sqcupm}{\ot}
\nc{\lepr}{\diamond_e} \nc{\vep}{\varepsilon} \nc{\labs}{\mid\!}
\nc{\rabs}{\!\mid} \nc{\hsha}{\widehat{\sha}}
\nc{\lsha}{\stackrel{\leftarrow}{\sha}}
\nc{\rsha}{\stackrel{\rightarrow}{\sha}} \nc{\lc}{\lfloor}
\nc{\rc}{\rfloor}
\nc{\tpr}{\sqcup}
\nc{\nctpr}{\vee}
\nc{\plpr}{\star}
\nc{\rbplpr}{\bar{\plpr}}
\nc{\sqmon}[1]{\langle #1\rangle}
\nc{\forest}{\calf}
\nc{\altx}{\Lambda_X} \nc{\vecT}{\vec{T}} \nc{\onetree}{\bullet}
\nc{\Ao}{\check{A}}
\nc{\seta}{\underline{\Ao}}
\nc{\deltaa}{\overline{\delta}}
\nc{\trho}{\tilde{\rho}}

\nc{\rpr}{\circ}
\nc{\dpr}{{\tiny\diamond}}
\nc{\rprpm}{{\rpr}}

\nc{\mmbox}[1]{\mbox{\ #1\ }} \nc{\ann}{\mrm{ann}}
\nc{\Aut}{\mrm{Aut}} \nc{\can}{\mrm{can}}
\nc{\twoalg}{{two-sided algebra}\xspace}
\nc{\colim}{\mrm{colim}}
\nc{\Cont}{\mrm{Cont}} \nc{\rchar}{\mrm{char}}
\nc{\cok}{\mrm{coker}} \nc{\dtf}{{R-{\rm tf}}} \nc{\dtor}{{R-{\rm
tor}}}
\renewcommand{\det}{\mrm{det}}
\nc{\depth}{{\mrm d}}
\nc{\Div}{{\mrm Div}} \nc{\End}{\mrm{End}} \nc{\Ext}{\mrm{Ext}}
\nc{\Fil}{\mrm{Fil}} \nc{\Frob}{\mrm{Frob}} \nc{\Gal}{\mrm{Gal}}
\nc{\GL}{\mrm{GL}} \nc{\Hom}{\mrm{Hom}} \nc{\hsr}{\mrm{H}}
\nc{\hpol}{\mrm{HP}} \nc{\id}{\mrm{id}} \nc{\im}{\mrm{im}}
\nc{\incl}{\mrm{incl}} \nc{\length}{\mrm{length}}
\nc{\LR}{\mrm{LR}} \nc{\mchar}{\rm char} \nc{\NC}{\mrm{NC}}
\nc{\mpart}{\mrm{part}} \nc{\pl}{\mrm{PL}}
\nc{\ql}{{\QQ_\ell}} \nc{\qp}{{\QQ_p}}
\nc{\rank}{\mrm{rank}} \nc{\rba}{\rm{RBA }} \nc{\rbas}{\rm{RBAs }}
\nc{\rbpl}{\mrm{RBPL}}
\nc{\rbw}{\rm{RBW }} \nc{\rbws}{\rm{RBWs }} \nc{\rcot}{\mrm{cot}}
\nc{\rest}{\rm{controlled}\xspace}
\nc{\rdef}{\mrm{def}} \nc{\rdiv}{{\rm div}} \nc{\rtf}{{\rm tf}}
\nc{\rtor}{{\rm tor}} \nc{\res}{\mrm{res}} \nc{\SL}{\mrm{SL}}
\nc{\Spec}{\mrm{Spec}} \nc{\tor}{\mrm{tor}} \nc{\Tr}{\mrm{Tr}}
\nc{\mtr}{\mrm{sk}}

\nc{\ab}{\mathbf{Ab}} \nc{\Alg}{\mathbf{Alg}}
\nc{\Algo}{\mathbf{Alg}^0} \nc{\Bax}{\mathbf{Bax}}
\nc{\Baxo}{\mathbf{Bax}^0} \nc{\RB}{\mathbf{RB}}
\nc{\RBo}{\mathbf{RB}^0} \nc{\BRB}{\mathbf{RB}}
\nc{\Dend}{\mathbf{DD}} \nc{\bfk}{{\bf k}} \nc{\bfone}{{\bf 1}}
\nc{\base}[1]{{a_{#1}}} \nc{\detail}{\marginpar{\bf More detail}
    \noindent{\bf Need more detail!}
    \svp}
\nc{\Diff}{\mathbf{Diff}} \nc{\gap}{\marginpar{\bf
Incomplete}\noindent{\bf Incomplete!!}
    \svp}
\nc{\FMod}{\mathbf{FMod}} \nc{\mset}{\mathbf{MSet}}
\nc{\rb}{\mathrm{RB}} \nc{\Int}{\mathbf{Int}}
\nc{\Mon}{\mathbf{Mon}}
\nc{\remarks}{\noindent{\bf Remarks: }}
\nc{\OS}{\mathbf{OS}} 
\nc{\Rep}{\mathbf{Rep}}
\nc{\Rings}{\mathbf{Rings}} \nc{\Sets}{\mathbf{Sets}}
\nc{\DT}{\mathbf{DT}}

\nc{\BA}{{\mathbb A}} \nc{\CC}{{\mathbb C}} \nc{\DD}{{\mathbb D}}
\nc{\EE}{{\mathbb E}} \nc{\FF}{{\mathbb F}} \nc{\GG}{{\mathbb G}}
\nc{\HH}{{\mathbb H}} \nc{\LL}{{\mathbb L}} \nc{\NN}{{\mathbb N}}
\nc{\QQ}{{\mathbb Q}} \nc{\RR}{{\mathbb R}} \nc{\BS}{{\mathbb{S}}} \nc{\TT}{{\mathbb T}}
\nc{\VV}{{\mathbb V}} \nc{\ZZ}{{\mathbb Z}}


\nc{\calao}{{\mathcal A}} \nc{\cala}{{\mathcal A}}
\nc{\calc}{{\mathcal C}} \nc{\cald}{{\mathcal D}}
\nc{\cale}{{\mathcal E}} \nc{\calf}{{\mathcal F}}
\nc{\calfr}{{{\mathcal F}^{\,r}}} \nc{\calfo}{{\mathcal F}^0}
\nc{\calfro}{{\mathcal F}^{\,r,0}} \nc{\oF}{\overline{F}}
\nc{\calg}{{\mathcal G}} \nc{\calh}{{\mathcal H}}
\nc{\cali}{{\mathcal I}} \nc{\calj}{{\mathcal J}}
\nc{\call}{{\mathcal L}} \nc{\calm}{{\mathcal M}}
\nc{\caln}{{\mathcal N}} \nc{\calo}{{\mathcal O}}
\nc{\calp}{{\mathcal P}} \nc{\calq}{{\mathcal Q}} \nc{\calr}{{\mathcal R}}
\nc{\calt}{{\mathcal T}} \nc{\caltr}{{\mathcal T}^{\,r}}
\nc{\calu}{{\mathcal U}} \nc{\calv}{{\mathcal V}}
\nc{\calw}{{\mathcal W}} \nc{\calx}{{\mathcal X}}
\nc{\CA}{\mathcal{A}}

\nc{\fraka}{{\mathfrak a}} \nc{\frakB}{{\mathfrak B}}
\nc{\frakb}{{\mathfrak b}} \nc{\frakd}{{\mathfrak d}}
\nc{\oD}{\overline{D}}
\nc{\frakF}{{\mathfrak F}} \nc{\frakg}{{\mathfrak g}}
\nc{\frakm}{{\mathfrak m}} \nc{\frakM}{{\mathfrak M}}
\nc{\frakMo}{{\mathfrak M}^0} \nc{\frakp}{{\mathfrak p}}
\nc{\frakS}{{\mathfrak S}} \nc{\frakSo}{{\mathfrak S}^0}
\nc{\fraks}{{\mathfrak s}} \nc{\os}{\overline{\fraks}}
\nc{\frakT}{{\mathfrak T}}
\nc{\oT}{\overline{T}}
\nc{\frakX}{{\mathfrak X}} \nc{\frakXo}{{\mathfrak X}^0}
\nc{\frakx}{{\mathbf x}}
\nc{\frakTx}{\frakT}      
\nc{\frakTa}{\frakT^a}        
\nc{\frakTxo}{\frakTx^0}   
\nc{\caltao}{\calt^{a,0}}   
\nc{\ox}{\overline{\frakx}} \nc{\fraky}{{\mathfrak y}}
\nc{\frakz}{{\mathfrak z}} \nc{\oX}{\overline{X}}

\font\cyr=wncyr10

\nc{\redtext}[1]{\textcolor{red}{#1}}


\title
[Bialgebras, CYBE and Manin triples for 3-Lie-algebras]
{Bialgebras, the classical Yang-Baxter equation and Manin triples for 3-Lie algebras}

\author{Chengming Bai}
\address{Chern Institute of Mathematics \& LPMC, Nankai University, Tianjin 300071, China}
         \email{baicm@nankai.edu.cn}

\author{Li Guo}
\address{Department of Mathematics and Computer Science,
         Rutgers University,
         Newark, NJ 07102}
\email{liguo@rutgers.edu}

\author{Yunhe Sheng}
\address{Department of Mathematics, Jilin University, Changchun 130012, Jilin, China} \email{shengyh@jlu.edu.cn}

\date{\today}

\begin{abstract}
This paper studies two types of 3-Lie bialgebras whose
compatibility conditions between the multiplication and
comultiplication are given by local cocycles and double
constructions respectively, and are therefore called the \typeI
and \typeII. They can be regarded as suitable extensions of the
well-known Lie bialgebra in the context of 3-Lie algebras, in two
different directions. The \typeI is introduced to extend the
connection between Lie bialgebras and the classical Yang-Baxter
equation. Its relationship with a ternary variation of the
classical Yang-Baxter equation, called the 3-Lie classical
Yang-Baxter equation, a ternary $\mathcal{O}$-operator and a
3-pre-Lie algebra is established. In particular, it is shown that
solutions of the 3-Lie classical Yang-Baxter equation give
(coboundary) \typeIs, whereas, 3-pre-Lie algebras give rise to
solutions of the 3-Lie classical Yang-Baxter equation. The \typeII
is introduced to extend to the 3-Lie algebra context the
connection between Lie bialgebras and double constructions of Lie
algebras. Their related Manin triples give a natural construction
of pseudo-metric 3-Lie algebras with neutral signature. Moreover,
the \typeII can be regarded as a special class of the \typeI.
Explicit examples of \typeIIs are provided.
\end{abstract}

\subjclass[2010]{16T10, 16T25, 17A30, 17B62, 81T30}

\keywords{Bialgebra, 3-Lie algebra, 3-Lie bialgebra, classical Yang-Baxter equation, $\mathcal{O}$-operator, Manin triple, local cocycle}

\maketitle
\vspace{-.5cm}

\tableofcontents

\numberwithin{equation}{section}

\allowdisplaybreaks

\section{Introduction}

\subsection{Bialgebras}

For a given algebraic structure determined by a set of
multiplications of various arities and a set of relations among the
operations (which can be made precise in the context of either
universal algebra or operads~\mcite{Gr,LV}), a bialgebra structure
on this algebra is obtained by a corresponding set of
comultiplications together with a set of compatibility conditions
between the multiplications and comultiplications.  For a finite
dimensional vector space $V$ with the given algebraic structure,
this can be achieved by equipping the dual space
$V^*$ with the same algebraic structure and a set of compatibility
conditions between the structures on $V$ and those on $V^*$.

 The associative bialgebra and infinitesimal bialgebra~\mcite{A1,JR} are
well-known bialgebra structures. Note that these two structures
have the same associative multiplications on $V$ and $V^*$. They
are distinguished only by the compatibility conditions, with the
comultiplication acting as a homomorphism (resp. a derivation) on
the multiplication for the associative bialgebra (resp.
the infinitesimal bialgebra). In general, it is quite common to have multiple bialgebra
structures that differ only by their compatibility conditions.

A good compatibility condition is prescribed on one hand by
a strong motivation and potential applications, and on the other
hand by a rich structure theory and effective constructions.

\subsection{Lie bialgebras}

In the Lie algebra context, the most common bialgebra structure
is
the Lie bialgebra, consisting of a Lie algebra
$(\g,[\cdot,\cdot]$) where $[\cdot,\cdot]:\otimes^2\g\to \g$ is a
Lie bracket, a Lie coalgebra $(\g,\Delta)$ where $\Delta:\g\to
\otimes^2\g$ is a Lie comultiplication, and a suitable
compatibility condition between the Lie bracket $[\cdot,\cdot]$
and the Lie comultiplication $\Delta$. The Lie bialgebra is the
algebraic structure corresponding to a Poisson-Lie group and the classical structure of a quantized
universal enveloping algebra~\mcite{CP,D}.
Such great
importance of the Lie bialgebra serves as the main motivation for
our interest in a suitable bialgebra theory for the 3-Lie algebra in this paper.

There are several equivalent statements for the compatibility
condition of a Lie bialgebra. It is these multiple manifestations
of the Lie bialgebra that determine its importance in both theory and applications. However the equivalence of similar conditions in the 3-Lie algebra context no longer holds.
Thus, for our purpose of developing a suitable bialgebra theory
for the 3-Lie algebra, we first differentiate the roles played by
these equivalent conditions.

For a Lie algebra $(\g,[\cdot,\cdot])$, there is a graded Lie
algebra structure on the exterior algebra $\Lambda^\bullet \g$.
The compatibility condition can be concisely stated as the condition that the
Lie comultiplication $\Delta$ is a derivation with respect to the
graded Lie algebra structure on $\Lambda^\bullet \g$:
 \begin{equation}
\Delta[x, y]=[\Delta x, y]+[x, \Delta y],\quad \forall x, y\in \g.
\mlabel{eq:liebider1}
\end{equation}
Note that, in fact,  $[x,\Delta y]=({\rm ad}_x\otimes 1+1\otimes
{\rm ad}_x)\Delta y$,  where
$\ad:\g\longrightarrow\gl(\g)$ is the adjoint representation.

If one goes beyond the simplicity of this definition of the
compatibility condition, one sees that the importance of Lie
bialgebra mainly comes from two other equivalent statements of the
compatibility condition.

One reason for the usefulness of the Lie bialgebra is
that it has a coboundary theory, which leads to the construction of Lie bialgebras from solutions of the classical Yang-Baxter equation.  This coboundary theory comes from the following equivalent statement of the compatibility condition
~(\mref{eq:liebider1}) as cocycles: the Lie algebra $\g$ acts on
$\otimes^2\g$  via the map $\ad\otimes 1+1\otimes \ad$ (i.e. the tensor representation of two adjoint
representations), and
$\Delta: \g\to \otimes^2\g$ is a $1$-cocycle on $\g$ with
coefficients in the representation $\ad\otimes 1+1\otimes \ad$.

On the other hand, some important applications of Lie
bialgebras to the related fields~\mcite{CP} have relied on the
Lie algebras with symmetric nondegenerate invariant bilinear forms
(the so-called ``self-dual" or ``quadratic" Lie algebras) coming
from Lie bialgebras. These Lie algebras are obtained from the
following Manin triple characterization of the Lie bialgebra: the compatibility condition~(\mref{eq:liebider1}) of
a Lie bialgebra   is given by the condition that the Lie algebras
$\g$ and $\g^*$ whose underlying vector space is the dual
space of $\g$ are subalgebras of a third Lie algebra $\g\oplus \g^*$ such that the
bilinear form
\begin{equation}\label{eq:liebidb}
( x+u^*, y+v^*)_+=\langle x, v^*\rangle+\langle u^*,y\rangle,\;\;\forall x,y\in \g, u^*,v^*\in \g^*,
\end{equation}
is invariant on $\g\oplus \g^*$. Here $\langle\cdot,\cdot \rangle$ is the usual pairing between $\g$ and $\g^*$.

To recap, we have the following three compatibility conditions of Lie bialgebras. They are all equivalent, yet each has its own advantages.

\begin{condition}
{\rm
\begin{enumerate}
\item\mlabel{it:lieder} the comultiplication $\Delta$ satisfies the
derivation condition in Eq.~(\ref{eq:liebider1});
\item the comultiplication $\Delta: \g\to
\otimes ^2 \g$ is a $1$-cocycle on $\g$; \mlabel{it:liecbd}
\item there
is a Manin triple $(\g\oplus \g^*,\g,\g^*)$. \mlabel{it:lietri}
\end{enumerate}
\label{cdn:liebialg} }
\end{condition}

\subsection{3-Lie algebras}
Generalizations of Lie algebras to higher arities, including 3-Lie
algebras and more generally, $n$-Lie
algebras~\mcite{Filippov,repKasymov,Ling}, have attracted attention
from several fields of mathematics and physics. It is the algebraic
structure corresponding to Nambu mechanics \mcite{ALMY,N,T}. In
particular, the study of 3-Lie algebras plays an important role in
string theory \mcite{BL1,M2branes,F1,G,HHM,Lorentzian3Lie,MFM1}. For
example, the structure of $3$-Lie algebras is applied to the study
of supersymmetry and gauge symmetry transformations of the
world-volume theory of multiple coincident M2-branes;  the
generalized identity for a $3$-Lie algebra is essential to define
the action with $N=8$ supersymmetry and the Jacobi identity can be
can be regarded as a generalized Pl\"ucker relation in the physics
literature.

Metric 3-Lie algebras are of particular interest in physics. More
precisely, to obtain the correct equations of motion for the
Begger-Lambert theory from a Lagrangian that is invariant under all
aforementioned symmetries seems to require the 3-Lie algebra to
admit an invariant inner product. The signature of this metric
determines the relative signs of the kinetic terms for scalar and
fermion fields in the Bagger-Lambert Lagrangian ~\mcite{BL1,BL2,G}.
In ordinary gauge theory, a positive-definite metric is required in
order to ensure that the theory has positive-definite kinetic terms
and to prevent violations of unitarity due to propagating ghost-like
degrees of freedom. However, there are few 3-Lie algebras which
admit positive-definite metrics. In fact, it has been shown
\mcite{GG,P} that all finite-dimensional real 3-Lie algebras with
positive-definite metrics are the direct sums of a special
4-dimensional real simple 3-Lie algebra and a trivial 3-Lie algebra.
On the other hand, in order to find new interesting Bagger-Lambert
Lagrangians, one is led to contemplating 3-Lie algebras with
pseudo-metrics having any signature $(p,q)$ or with degenerate
invariant symmetric bilinear forms, despite the possibility of
negative-norm states, since in certain dynamical systems a zero-norm
generator corresponds to a gauge symmetry while a negative norm
generator corresponds to a ghost. Thus it seems worthwhile and
interesting, from both physical and mathematical considerations, to
find new 3-Lie algebras with symmetric invariant
bilinear forms.

\subsection{3-Lie bialgebras}
Given the importance of Lie bialgebras and 3-Lie algebras, it is
natural to develop a suitable bialgebra theory for 3-Lie algebras.

Motivated by the equivalent compatibility conditions of Lie
bialgebras in Condition~\ref{cdn:liebialg}, one is naturally led to
defining a 3-Lie bialgebra as a pair consisting of a 3-Lie algebra
$(A,[\cdot,\cdot,\cdot])$, a 3-Lie coalgebra $(A,\Delta)$ such that
$(A^*,\Delta^*)$ is also a 3-Lie algebra and one of the following
compatibility conditions is satisfied.

\begin{condition}

{\rm
\begin{enumerate}

\item\mlabel{it:der} the comultiplication $\Delta$ satisfies certain
``derivation" condition; \item the comultiplication $\Delta: A\to
\otimes ^3 A$ is a $1$-cocycle on $A$; \mlabel{it:cbd} \item there
is a Manin triple $(A\oplus A^*,A,A^*)$. \mlabel{it:tri}
\end{enumerate}
\label{cdn:3lie}
}
\end{condition}

Contrary to the case of the Lie bialgebra, suitable extensions of
these conditions to the context of 3-Lie algebras are not
equivalent, leading to different extensions of the Lie bialgebra to
the ternary case. Thus there might not be a unique ``perfect"
definition of a 3-Lie bialgebra, but three different versions
serving different purposes. This is in reminiscent to the case of
bialgebras. The following is an overview of the three approaches and
also serves as an outline of this paper. The first approach was
given in ~\mcite{Bai3Liebi} based on
Condition~\mref{cdn:3lie}.~(\mref{it:der}). The other two are the
subjects of study of this paper.

\subsubsection{$3$-Lie bialgebras with the derivation compatibility}
An approach of bialgebra theory for 3-Lie algebras based on
Condition~\ref{cdn:3lie}.(\mref{it:der}) was taken in \cite{Bai3Liebi}, in
which the authors generalized the compatibility condition
\eqref{eq:liebider1} formally to the following equality:
\begin{equation}
\Delta[x,y,z]=[\Delta x, y, z]+[x, \Delta y, z] +[x, y, \Delta z],
\mlabel{eq:3lieder}
\end{equation}
where \begin{equation} \mlabel{eq:3-der} [x,y,\Delta z]=-[x,\Delta
z, y]=[\Delta z, x,y]= (\ad_{x,y}\otimes 1\otimes 1+1\otimes
\ad_{x,y}\otimes 1 +1\otimes 1\otimes \ad_{x,y})\Delta(z),
\end{equation}
with $\ad_{x,y}(z)=[x,y,z]$
 for all $x,y,z\in A$.

However, for this formal generalization, neither a coboundary
theory nor the structure   on the double space $A\oplus A^*$ is
known. Unlike the case of Lie
algebras, it is still unknown whether there is 3-Lie algebra
structure on $\Lambda^\bullet A$, making it challenging for
such a bialgebra structure to develop an expected structure theory and applications.

\subsubsection{$3$-Lie bialgebras with cocycle compatibility}

Since a large part of the Lie bialgebra theory is based on the
cocycle characterization of its compatibility condition, it is
natural to extend this approach to 3-Lie algebras via
Condition~\ref{cdn:3lie}.(\mref{it:cbd}). We will carry out this approach in Section~\ref{sec:cycle}, after some preparation in Section~\mref{sec:bas}.
 Unfortunately, unlike the case of the Lie bialgebra,  such a cocycle description by itself
does not make sense in the 3-Lie algebra context since there is no
natural representation of a 3-Lie algebra $A$ on $\otimes^3 A$. To
get around this obstacle, we observe that, for a Lie bialgebra
$(\g,\Delta)$, there are two more 1-cocycles $\Delta_1$ and
$\Delta_2$  associated to the representations $(\otimes^2 \g, {\rm
ad}\otimes 1)$  and $(\otimes^2 \g, 1\otimes \ad)$ respectively and
the cocycle condition for $\Delta: \g\to \otimes^2\g$ follows from
the assumption that $\Delta$ is a sum $\Delta_1+\Delta_2$ of
1-cocycles with certain additional constrains. Noting that
$(\otimes^3 A,\ad\otimes 1\otimes 1)$, $(\otimes^3 A, 1\otimes
\ad\otimes 1)$ and $(\otimes^3 A, 1\otimes 1\otimes \ad)$ are
representations of a 3-Lie algebra $A$,
we are naturally led to the following modification of Condition~\ref{cdn:3lie}.(\mref{it:cbd}):
\begin{enumerate}
\item[(\ref{it:cbd}')] The comultiplication $\Delta:A\to
\otimes^3A$ is a sum $\Delta_1+\Delta_2+\Delta_3$, where each
$\Delta_i$ for $i=1,2,3$, is a $1$-cocycle associated to
$(\otimes^3 A, \ad\otimes 1\otimes 1)$, $(\otimes^3 A, 1\otimes
\ad\otimes 1)$ and $(\otimes^3 A, 1\otimes 1\otimes \ad)$
respectively.
\end{enumerate}

With this definition, we obtain a bialgebra structure for the 3-Lie
algebra that also has a coboundary theory. We call this structure
a local cocycle 3-Lie bialgebra (Definition~\mref{defi:typeI}). Although there is no natural
3-Lie algebra structure on $A\oplus A^*$, there are still some
interesting properties for local cocycle 3-Lie bialgebras. In
particular, it also leads to an analogue of the classical Yang-Baxter
equation defined on a 3-Lie algebra, called the 3-Lie classical Yang-Baxter equation (Definition~\mref{defi:3cybe}). Solutions of this equation give natural constructions of local cocycle 3-Lie bialgebras (Theorem~\mref{thm:ybe}).  On the
other hand, such a bialgebra theory is also related to $\mathcal
O$-operators and the so-called 3-pre-Lie-algebras which provide solutions of the 3-Lie classical Yang-Baxter equation
and hence constructions of local cocycle 3-Lie bialgebras on
certain double spaces. This is done in Section~\mref{ss:oop}.

\subsubsection{$3$-Lie bialgebras with the Manin triple compatibility}

Motivated by the double structure of a Manin triple characterizing
a Lie bialgebra, it is natural to consider the following analogue
of a Manin triple for 3-Lie algebras based on
Condition~\ref{cdn:3lie}.(\mref{it:tri}):
\begin{enumerate}
\item[(\ref{it:tri}')] $A\oplus A^*$ can
be equipped with a 3-Lie algebra structure such that $A$ and $A^*$
are 3-Lie subalgebras of $A\oplus A^*$ and the bilinear form on
$A\oplus A^*$ defined by Eq.~(\ref{eq:liebidb}) is invariant.
\end{enumerate}
Note that if such a structure exists, then $(A\oplus A^*,
(\cdot,\cdot)_+ )$ is a 3-Lie algebra with a pseudo-metric having
the signature $(n,n)$, where $n:=\dim A$. Thus, such an approach
provides a natural construction of pseudo-metric 3-Lie algebras with
signature $(n,n)$ for the aforementioned study of Bagger-Lambert
Lagrangians. Following the theory of Lie bialgebras, we utilize a
Manin triple on a 3-Lie algebra to define a bialgebra structure on a
3-Lie algebra, called a \typeII. We will present this approach in
Section~\mref{sec:triple}. We also show that the
\typeII can be regarded as a special case of the \typeI and provide
explicit examples of \typeIIs.

This completes the overview of the paper. A more detailed summary of each section can be found at the beginning of the section.
Moreover, throughout this paper, all the vector spaces and algebras are assumed to be of finite dimension, although
many results still hold in the infinite dimensional cases.

\section{Some preliminary results on 3-Lie algebras}
\label{sec:bas}

In this section, we give some general results on 3-Lie algebras and their cohomology theory that will be used in later sections.

\begin{defi}{\rm\cite{Filippov}}
  A {\bf 3-Lie algebra}  is a vector space $A$ together with a skew-symmetric linear map ($3$-Lie bracket) $[\cdot,\cdot,\cdot]:
\otimes^3 A\rightarrow A$ such that the following {\bf Fundamental Identity (FI)} holds:
\begin{equation}\label{eq:de1}
[x_1,x_2,[x_3,x_4,x_5]]=[[x_1,x_2,x_3],x_4,x_5]+[x_3,[x_1,x_2,x_4],x_5]+[x_3,x_4,[x_1,x_2,x_5]]
\end{equation}
for $x_i\in A, 1\leq i\leq 5$.
\end{defi}
In other words, for $x_1, x_2\in A$, the operator
\begin{equation}\label{eq:adjoint}
\ad_{x_1,x_2}:A\to A, \quad \ad_{x_1,x_2}x:=[x_1,x_2,x], \quad \forall x\in A,
\end{equation}
is a derivation in the sense that
$$ \ad_{x_1,x_2}[x_3,x_4,x_5]=[\ad_{x_1,x_2}x_3,x_4,x_5] +[x_3,\ad_{x_1,x_2}x_4,x_5]+[x_3,x_4,\ad_{x_1,x_2}x_5], \forall x_1, x_2, x_3\in A.$$

A morphism between 3-Lie algebras is defined as usual, i.e. a linear map that preserves the 3-Lie brackets.

\begin{pro}\label{pro:someequalities}
Let $A$ be a vector space  together with a skew-symmetric linear map $[\cdot,\cdot,\cdot]:
\otimes^3 A\rightarrow A$. Then $(A,[\cdot,\cdot,\cdot])$ is a $3$-Lie algebra if and only if the following identities hold:
\begin{enumerate}
\item $[[x_1,x_2,x_3],x_4, x_5]-[[x_1,x_2,x_4],x_3,x_5]+[[x_1,x_3,x_4],x_2,x_5]-[[x_2,x_3,x_4],x_1,x_5]=0$,
\item $[[x_1,x_2,x_5],x_3,x_4]+[[x_3,x_4,x_5],x_1,x_2]-[[x_1,x_3,x_5],x_2,x_4]\\
-[[x_2,x_4,x_5], x_1,x_3]+[[x_1,x_4,x_5],x_2,x_3]+[[x_2,x_3,x_5],x_1,x_4]=0$,
\end{enumerate}for  $x_i\in A, 1\leq i\leq 5$.\label{pro:new}
\end{pro}

\pf
If $(A,[\cdot,\cdot,\cdot])$ is a 3-Lie algebra, then applying Eq.~(\ref{eq:de1}) to the last term of Eq.~(\ref{eq:de1}), we have
\begin{eqnarray*}
  [x_1,x_2,[x_3,x_4,x_5]]
  &=&[[x_1,x_2,x_3],x_4,x_5]+[x_3,[x_1,x_2,x_4],x_5]\\
  &&+[[x_3,x_4,x_1],x_2,x_5]+[x_1,[x_3,x_4,x_2],x_5]+[x_1,x_2,[x_3,x_4,x_5]],
\end{eqnarray*}
proving Item (a).

Similarly, applying Eq.~(\ref{eq:de1}) to the first and second terms on the right hand side of Eq.~(\ref{eq:de1}), we obtain
\begin{eqnarray*}
  [x_1,x_2,[x_3,x_4,x_5]]
 &=& [[x_4,x_5, x_1],x_2,x_3]+[x_1,[x_4,x_5,x_2],x_3] +[x_1,x_2,[x_4,x_5,x_3]]\\
 &&+[[x_5, x_3,x_1],x_2,x_4]+[x_1,[x_5, x_3,x_2],x_4]+[x_1,x_2,[x_5, x_3,x_4]]\\
 &&+[x_3,x_4,[x_1,x_2,x_5]],
\end{eqnarray*}
implying Item (b).
\smallskip

Conversely, suppose that Items~(a) and (b) hold. First Item~(a) gives
$$
[[x_1,x_2,x_3],x_4,x_5]=-[[x_1,x_2,x_5],x_3,x_4]+[[x_1,x_3,x_5],x_2,x_4]-[[x_2,x_3,x_5],x_1,x_4]$$
which equals to
$$[[x_3,x_4,x_5],x_1,x_2]-[[x_2,x_4,x_5],x_1,x_3]+[x_1,x_4,x_5],x_2,x_3]$$
by Item~(b).
Thus, $A$ is a 3-Lie algebra.
\qed\vspace{3mm}

The notion of a representation of an $n$-Lie algebra was introduced in \cite{repKasymov}. See also \cite{rep}.
\begin{defi}\label{defi:rep}
 Let $V$ be a vector space. A {\bf representation of a 3-Lie algebra} $A$ on $V$ is a skew-symmetric linear map $\rho: \otimes^2A\rightarrow \gl(V)$ such that
\begin{enumerate}
\item [\rm(i)] $\rho (x_1,x_2)\rho(x_3,x_4)-\rho(x_3,x_4)\rho(x_1,x_2)=\rho([x_1,x_2,x_3],x_4)-\rho([x_1,x_2,x_4],x_3)$;
\item [\rm (ii)]$\rho ([x_1,x_2,x_3],x_4)=\rho(x_1,x_2)\rho(x_3,x_4)+\rho(x_2,x_3)\rho(x_1,x_4)+\rho(x_3,x_1)\rho(x_2,x_4)$,
\end{enumerate}
for $x_i\in A, 1\leq i\leq 4$.
\end{defi}

It is straightforward to obtain
\begin{lem}
Let $A$ be a $3$-Lie algebra, $V$  a vector space and $\rho:
\otimes^2A\rightarrow \gl(V)$  a skew-symmetric linear
map. Then $(V,\rho)$ is a representation of $A$ if and only if there
is a $3$-Lie algebra structure $($called the semi-direct product$)$
on the direct sum $A\oplus V$ of vector spaces, defined by
\begin{equation}\label{eq:sum}
[x_1+v_1,x_2+v_2,x_3+v_3]_{A\oplus V}=[x_1,x_2,x_3]+\rho(x_1,x_2)v_3+\rho(x_3,x_1)v_2+\rho(x_2,x_3)v_1,
\end{equation}
for $x_i\in A, v_i\in V, 1\leq i\leq 3$. We denote this semi-direct product $3$-Lie algebra by $A\ltimes_\rho V.$
\end{lem}

From the proof of Proposition~\ref{pro:new}, we immediately obtain
\begin{pro}\label{pro:rep properties}
Let $(V,\rho)$ be a representation of a $3$-Lie algebra $A$. Then the following identities hold:
\begin{enumerate}\label{pro:new-repn}
\item $\rho([x_1,x_2,x_3],x_4)-\rho([x_1,x_2,x_4],x_3])+\rho([x_1,x_3,x_4],x_2)-\rho([x_2,x_3,x_4],x_1)=0$;
\item $\rho(x_1,x_2)\rho(x_3,x_4)+\rho(x_2,x_3)\rho(x_1,x_4)+\rho(x_3,x_1)\rho(x_2,x_4)\\
+\rho(x_3,x_4)\rho(x_1,x_2)+\rho(x_1,x_4)\rho(x_2,x_3)+\rho(x_2,x_4)\rho(x_3,x_1)=0$,
\end{enumerate}
for $x_i\in A, 1\leq i\leq 4$.
\end{pro}

Let $(V,\rho)$ be a representation of a $3$-Lie algebra $A$. Define $\rho^*:\otimes^2 A\longrightarrow\gl(V^*)$
by
\begin{equation}\label{eq:dual}
\langle \rho^*(x_1,x_2) \alpha, v\rangle =-\langle \alpha,\rho(x_1,x_2) v\rangle,\quad \forall \alpha\in V^*,~x_1,x_2\in A,~ v\in V.
\end{equation}
A straightforward computation applying Definition~\ref{defi:rep} and Proposition~\ref{pro:rep properties} gives
\begin{pro}
With the above notations, $(V^*,\rho^*)$ is a representation of $A$, called the dual representation.
\end{pro}

\begin{ex}{\rm Let $A$ be a 3-Lie algebra. The linear map ${\rm ad}:\otimes^2A\rightarrow \frak g\frak l(A)$ with $x_1,x_2\rightarrow {\rm ad}_{x_1,x_2}$ for any $x_1,x_2\in A$
defines a representation $(A, {\rm ad})$ which is called the {\bf adjoint representation} of $A$, where ${\rm ad}_{x_1,x_2}$ is given by Eq.~\eqref{eq:adjoint}.
The dual representation $(A^*,{\rm ad}^*)$ of the adjoint representation $(A,{\rm ad})$ of a 3-Lie algebra $A$ is called
the {\bf coadjoint representation}.

}\end{ex}

Given a representation $(V,\rho)$, denote by $ C^p(A;V)$ the set of $p$-cochains:
$$
C^p(A;V):=\{\mbox{linear~ maps}~ f:\underbrace{\wedge^2A\otimes\cdots\otimes \wedge^2A}_{p-1 \text{ factors}}\wedge A\longrightarrow V \}.
$$
The coboundary operators associated to the trivial representation and the adjoint representation are given in \cite{cohomology,review}. Similarly, using the notation $X_j:=X_j^1\wedge X_j^2\in\wedge^2A, 1\leq j\leq p$, the coboundary operator $\delta:C^p(A;V)\longrightarrow C^{p+1}(A;V)$ is defined by
\begin{eqnarray*}
 &&\delta f(X_1,\cdots,X_{p},Z)\\
 &=&\sum_{1\leq j\leq k}^{p}(-1)^jf(X_1,\cdots,\widehat{X_j},\cdots,X_{k-1},[X_j^1,X_j^2,X_k^1]\wedge X_k^2+X_k^1\wedge [X_j^1,X_j^2,X_k^2],\cdots,X_{p},Z)\\
 &&+\sum_{j=1}^{p}(-1)^jf(X_1,\cdots,\widehat{X_j},\cdots,X_{p},[X_j^1,X_j^2,Z])\\
 &&+\sum_{j=1}^{p}(-1)^{j+1}\rho(X_j^1, X_j^2)f(X_1,\cdots,\widehat{X_j},\cdots,X_{p},Z)\\
 &&+(-1)^{p+1}\rho(X_{p}^2,Z)f(X_1,\cdots,X_{p},X_{p}^1)+(-1)^{p}\rho(X_{p}^1,Z)f(X_1,\cdots,X_{p},X_{p}^2).
\end{eqnarray*}
In particular, we  obtain the following formula for a 1-cocycle.

\begin{defi}
Let $A$ be a $3$-Lie algebra and $(V,\rho)$ be a
representation of $A$.  A linear map $f:A\rightarrow V$ is called
a {\bf $1$-cocycle} on $A$ associated to $(V,\rho)$ if it satisfies
\begin{equation}
f([x_1,x_2,x_3])=\rho(x_1,x_2)f(x_3)+\rho(x_2,x_3)f(x_1)+\rho(x_3,x_1)f(x_2),\;\;\forall x_1,x_2,x_3\in A.\label{eq:1-cocycle}
\end{equation}
\end{defi}

\begin{ex}\mlabel{ex:d}
{\rm Recall that a {\bf derivation} $D$ of a 3-Lie algebra $A$ is defined to be a linear map $D:A\rightarrow A$ satisfying
$$D[x_1,x_2,x_3]=[D(x_1),x_2,x_3]+[x_1,D(x_2),x_3]+[x_1,x_2,D(x_3)],\quad\forall x_1,x_2,x_3\in A.$$
So any derivation is a 1-cocycle of the 3-Lie algebra $A$
associated to the adjoint representation $(A,{\rm
ad})$ given by Eq.~\eqref{eq:adjoint}.
In particular, there is a 1-cocycle (derivation) depending
on two elements in the representation space $A$ which can be
regarded as a kind of ``1-coboundary". Explicitly, for two fixed
elements $u,v\in A$, the linear map $D:A\rightarrow A$ defined by
$$D(x)={\rm ad} (u,v)x=[u,v,x],\quad\forall x\in A,$$
is a 1-cocycle on $A$ associated to the adjoint representation $(A, {\rm ad})$. }\end{ex}

\section{Local cocycle 3-Lie bialgebras}
\label{sec:cycle}
We study \typeIs in this section. In Section 3.1, we introduce the notion of a \typeI. In Section 3.2, we study
coboundary \typeIs and introduce the notion
of the 3-Lie classical Yang-Baxter equation. Given a solution of
the 3-Lie classical Yang-Baxter equation, we can construct a local
cocycle 3-Lie bialgebra (Theorem \ref{thm:ybe}). In Section
3.3, we introduce the notion of an $\mathcal O$-operator on a
3-Lie algebra, which gives a skew-symmetric solution of the
3-Lie classical Yang-Baxter equation in some semi-direct product
3-Lie algebra (Theorem \ref{thm:O-r}). Then we introduce the
notion of a 3-pre-Lie algebra, and apply it to obtain solutions of the
3-Lie classical Yang-Baxter equation in some semi-direct product
3-Lie algebra associated to a 3-pre-Lie algebra.

\subsection{Local cocycles and bialgebras}

We first give the definition of a \typeI.

\begin{defi}
A {\bf \typeI} is a pair $(A,\Delta)$, where $A$ is a $3$-Lie algebra and $\Delta=\Delta_1+\Delta_2+\Delta_3:A\rightarrow A\otimes
A\otimes A$ is a linear map, such that $\Delta^*:A^*\otimes A^*\otimes A^*\rightarrow
A^*$ defines a $3$-Lie algebra structure on
$A^*$, and the following conditions are satisfied:
\begin{itemize}
\item $\Delta_1$ is a $1$-cocycle associated to the representation $(A\otimes A\otimes A,{\rm ad}\otimes 1\otimes 1)$;
\item $\Delta_2$ is a $1$-cocycle associated to the representation $(A\otimes A\otimes A, 1\otimes {\rm ad}\otimes 1)$;
\item $\Delta_3$ is a $1$-cocycle associated to the representation $(A\otimes A\otimes A, 1\otimes 1\otimes {\rm ad})$.
\end{itemize}
\mlabel{defi:typeI}
\end{defi}

A \typeI is a natural generalization of a Lie
bialgebra. Recall that a Lie bialgebra is a Lie algebra $\frak g$
with a linear map $\Delta:\frak g\rightarrow \frak g\otimes \frak g$
such that $\Delta$ defines a Lie co-algebra and
\begin{equation}\label{eq:conditionbi}
\Delta([x,y])=({\rm ad}_x\otimes 1+1\otimes {\rm ad}_x)\Delta (y)-
({\rm ad}_y\otimes 1+1\otimes {\rm ad}_y)\Delta (x),\quad\forall ~ x,y\in
\frak g.\end{equation}
 The usual interpretation is that $\Delta$ is a 1-cocycle
of $\frak g$ associated to the   representation $ {\rm ad}\otimes
1+1\otimes {\rm ad}$ on the tensor space $\frak g\otimes \frak g$.
Note that for a Lie algebra, due to the underlying enveloping
algebra $U(\frak g)$ being a Hopf algebra,  there is a natural
representation on the tensor space. While this fact cannot be
extended to 3-Lie algebras, the fact that both $(\frak g\otimes \frak g, {\rm
ad}\otimes 1)$ and $(\frak
g\otimes \frak g, 1\otimes {\rm ad})$ are representations of $\frak
g$ can be extended, leading  to the concept of a \typeI.

There is an alternative interpretation of a Lie bialgebra.
\begin{lem}\label{lem:Lie bi}
  A pair $(\g,\Delta)$ is a Lie bialgebra if
$\Delta=\Delta_1+\Delta_2$ such that for any $x,y\in \g$,
 \begin{enumerate}
\item $\Delta_1$ and $\Delta_2$ are $1$-cocycles of $\frak g$ associated to  ${\rm
ad}\otimes 1$ and $1\otimes {\rm ad}$ respectively, i.e.
\begin{equation}\label{eq:conditionred}
\Delta_1[x,y]=({\rm ad}_x\otimes 1) \Delta_1(y)-({\rm ad}_y\otimes1) \Delta_1(x), \quad \Delta_2[x,y]=({1\otimes \rm ad}_x) \Delta_2(y)-(1\otimes {\rm ad}_y)\Delta_2(x);
\end{equation}
\item the following compatibility condition holds:
\begin{equation}\label{eq:conditionextra}
({\rm ad}_x\otimes 1)\Delta_2(y)+(1\otimes {\rm
ad}_x)\Delta_1(y)- ({\rm ad}_y\otimes1) \Delta_2(x)-(1\otimes {\rm
ad}_y)\Delta_1(x)=0.
\end{equation}
\end{enumerate}
\end{lem}
\pf It follows from the fact that   Eqs.~\eqref{eq:conditionred} and~\eqref{eq:conditionextra} imply Eq.~\eqref{eq:conditionbi}.\qed\vspace{3mm}

In   the theory of Lie bialgebras, it is essential to consider the
coboundary case, which is related to the theory of the classical
Yang-Baxter equation. In the coboundary case, we have
$\Delta(x)=({\rm ad}_x\otimes 1+1\otimes {\rm ad}_x)r$,  for a fixed $r\in
\frak g\otimes \frak g$ and any $x\in \g$.
In view of Lemma \ref{lem:Lie bi},
 it is
natural to take
\begin{equation}\label{eq:coboundary}
\Delta_1(x)=({\rm ad}_x\otimes 1) (r),\;\;\Delta_2(x)=(1\otimes {\rm
ad}_x)(r),\quad \forall x\in \g.
\end{equation}
Under this condition,  Eq.~\eqref{eq:conditionextra} holds
automatically by a straightforward computation. Thus, for the
purpose of the classical Yang-Baxter equation, it is enough to only
require Condition (a) in Lemma \ref{lem:Lie bi}. The compatibility
condition for \typeIs is a
natural extension of this condition. Note that Condition~(a) alone
cannot guarantee that $\g\oplus \g^*$ is a Lie algebra.

We end this subsection with an interpretation of the
Condition~(a) in Lemma \ref{lem:Lie bi} from an operadic point of
view. See~\cite{Lo} for the background and notations. The
compatibility condition of a Lie bialgebra can be expressed as
\begin{equation}
\Delta[x,y]=x_{(1)}\ot [x_{(2)},y] + [x_{(1)},y]\ot x_{(2)} + [x,y_{(1)}]\ot y_{(2)} + y_{(1)}\ot [x,y_{(2)}]
\label{eq:cmp}
\end{equation}
in Sweedler's notation.
Even though it does not fit in the frame of the generalized bialgebra of Loday in the sense that it gives a good triple in~\cite{Lo}, it has a ``unitarization" called {\bf Lie$^c$-Lie-bialgebra} that does. Its compatibility condition is $$
\Delta[x,y]=2(x\otimes y-y\otimes x) + \frac{1}{2}(x_{(1)}\ot [x_{(2)},y] + [x_{(1)},y]\ot x_{(2)} + [x,y_{(1)}]\ot y_{(2)} + y_{(1)}\ot [x,y_{(2)}]).
$$
In a similar way, a half of the compatibility condition in Eq.~(\ref{eq:cmp})
$$ \Delta[x,y] = [x_{(1)},y]\ot x_{(2)} + [x,y_{(1)}]\ot y_{(2)}$$
has the unitarization

$$ \Delta[x,y]= x\ot y + [x_{(1)},y]\ot x_{(2)} + [x,y_{(1)}]\ot y_{(2)}.$$
This is the compatibility condition of the co-variation of the {\bf NAP}$^c$-{\bf PreLie}-bialgebra of Livernet which is a good triple.
Likewise, the other half of Eq.~(\ref{eq:cmp})
$$\Delta[x,y] = x_{(1)}\ot [x_{(2)},y] + y_{(1)}\ot [x,y_{(2)}] $$
has its unitarization the compatibility condition of another good triple.

In fact, we can put the four terms of Eq.~(\ref{eq:cmp}) in a diagram
$$
\xymatrix{ x_{(1)}\ot [x_{(2)},y] \ar@{-}[rr]^{\text{opp Liv}} \ar@{-}[dd]_{\text{inf}} \ar@{-}[rrdd]^(.3){\Delta_2}&& [x_{(1)},y]\ot x_{(2)} \ar@{-}[dd]^{\text{opp inf}} \\
&& \\
[x,y_{(1)}]\ot y_{(2)} \ar@{-}[rr]^{\text{Liv}} \ar@{-}[rruu]^(.3){\Delta_1} &&  y_{(1)}\ot [x,y_{(2)}].}
$$
Then the sum of the left (resp. right) two terms is the compatibility condition of the infinitesimal (resp. opposite infinitesimal) operads. The sum of the bottom (resp. top) two terms are compatibility of the Livernet (resp. opposite Livernet) operad. The sum of the diagonal (resp. opposite diagonal) two terms are for $\Delta_1$ (resp. $\Delta_2$) in our discussion in Eq.~(\ref{eq:conditionextra}).

\subsection{Coboundary \typeIs and the \tcybe}
\mlabel{ss:3cybe}
In this subsection, we study coboundary local cocycle 3-Lie bialgebras, i.e. construct a \typeI from an element $r\in A\otimes A$.
First we give some preliminary notations.

Let $A$ be a vector space. For any $T=x_1\otimes x_2\otimes
\cdots \otimes x_n\in \otimes^n A $ and $1\leq i<j\leq n$, define the $(ij)$-switching operator
$$\sigma_{ij} (T)=x_{1}\otimes \cdots\otimes x_{j}\otimes \cdots\otimes x_i\otimes\cdots
\otimes x_n.$$

For any $1\leq p\neq q\leq n$, define an inclusion $\cdot_{pq}:\otimes^2A\longrightarrow \otimes^n A$ by sending $r=\sum_i x_i\otimes y_i\in A\otimes A$ to
$$
r_{pq}:=\sum_i z_{i1}\otimes\cdots\otimes z_{in},\quad \text{ where } z_{ij}=\left\{\begin{array}{ll} x_i,& j=p,\\ y_i, & j=q, \\ 1, & i\neq p, q. \end{array} \right.
$$
In other words, $r_{pq}$ puts $x_i$ at the $p$-th position, $y_i$
at the $q$-th position and 1 elsewhere in an $n$-tensor, where 1
is the unit if $A$ is a unital algebra, otherwise,
1 is a symbol playing a similar role of unit.
For example, when $n=4$, we have
\begin{eqnarray*}
&&r_{12}=\sum_i x_i\otimes y_i\otimes 1\otimes 1\in A^{\otimes 4},\quad
r_{21}=\sum_i y_i\otimes x_i\otimes 1\otimes 1\in A^{\otimes 4}.
\end{eqnarray*}

When $A$ is a 3-Lie algebra with the 3-Lie bracket
$[\cdot,\cdot,\cdot]$, for any $r=\sum_ix_i\otimes y_i\in A\otimes
A$, we define $[[r,r,r]]\in \otimes^4 A$ by
\begin{eqnarray}
\label{eq:rrr}[[r,r,r]]&:=&[r_{12},r_{13},r_{14}]+[r_{12},r_{23},r_{24}]+[r_{13},r_{23},r_{34}]+[r_{14},r_{24},r_{34}]\\
\nonumber&=&\sum_{i,j,k}\big([x_i,x_j,x_k]\otimes y_i\otimes y_j\otimes y_k+x_i\otimes [y_i,x_j,x_k]\otimes y_j\otimes y_k\\
\nonumber&&+ x_i\otimes x_j\otimes [y_i, y_j,x_k]\otimes y_k+ x_i\otimes x_j\otimes x_k\otimes [y_i,y_j,y_k]\big).
\end{eqnarray}

For any $r=\sum_ix_i\otimes y_i\in A\otimes A$, set
\begin{equation}\label{eq:delta123}\left\{\begin{array}{ccc}
\Delta_1(x)&:=&\sum_{i,j} [x,x_i,x_j]\otimes y_j\otimes y_i;\\
\Delta_2(x)&:=&\sum_{i,j} y_i\otimes [x,x_i,x_j]\otimes y_j;\\
\Delta_3(x)&:=&\sum_{i,j} y_j\otimes y_i\otimes [x,x_i,x_j],
\end{array}\right.
\end{equation}
where $x\in A$.

\begin{lem}
With the above notations, we have
\begin{enumerate}
\item[\rm (i)] $\Delta_1$ is a $1$-cocycle associated to the representation $(A\otimes A\otimes A, {\rm ad}\otimes 1\otimes 1)$;
\item[\rm (ii)] $\Delta_2$ is a $1$-cocycle associated to the representation $(A\otimes A\otimes A,1\otimes {\rm ad}\otimes 1)$;
\item[\rm (iii)] $\Delta_3$ is a $1$-cocycle associated to the representation $(A\otimes A\otimes A,1\otimes 1\otimes {\rm ad})$.
\end{enumerate}
\end{lem}

\begin{proof}
For all $x,y,z\in A$, we have
\begin{eqnarray*}
\Delta_1([x,y,z])&=& \sum_{i,j} [[x,y,z],x_i,x_j]\otimes y_j\otimes y_i\\
&=&\sum_{i,j} \left([[x,x_i,x_j],y,z]+[x,[y,x_i,x_j],z]+[x,y,[z,x_i,x_j]]\right)\otimes y_j\otimes y_i\\
&=&({\rm ad}_{y,z}\otimes 1\otimes 1)\Delta_1(x)+({\rm ad}_{z,x}\otimes 1\otimes 1)\Delta_1(y)+({\rm ad}_{x,y}\otimes 1\otimes 1)\Delta_1(z).
\end{eqnarray*}
Therefore, $\Delta_1$ is 1-cocycle  associated to the representation $(A\otimes A\otimes A,{\rm ad}\otimes 1\otimes 1)$. The other two statements can be proved similarly.
\end{proof}

\begin{pro}\label{pro:skew-symmetric}
Let $A$ be a $3$-Lie algebra and $r\in A\otimes A$.
Let $\Delta=\Delta_1+\Delta_2+\Delta_3$, where
$\Delta_1,\Delta_2,\Delta_3$ are induced by $r$ as in Eq.~\eqref{eq:delta123}. Then
$\Delta^*:A^*\otimes A^*\otimes A^*\rightarrow A^*$ defines a
skew-symmetric operation.
\end{pro}

\begin{proof}
We only need to prove that for all $x\in A$,
$$\Delta(x)+\sigma_{12}
\Delta(x)=0,\;\Delta(x)+\sigma_{23}\Delta(x)=0.$$ In fact, we have
\begin{eqnarray*}
\sigma_{12}\Delta_1(x)&=&\sum_{i,j}y_j\otimes [x,x_i,x_j]\otimes
y_i=\sum_{i,j}y_i\otimes [x,x_j,x_i]\otimes y_j=-\Delta_2(x);\\
\sigma_{12}\Delta_2(x)&=& \sum_{i,j}[x,x_i,x_j]\otimes y_i\otimes
y_j=-\Delta_1(x);\\
\sigma_{12}\Delta_3(x)&=& \sum_{i,j}y_i\otimes y_j\otimes
[x,x_i,x_j]=-\Delta_3(x).
\end{eqnarray*}
Hence $\sigma_{12} \Delta(x)=-\Delta(x)$. Similarly, we have
$\sigma_{23}
\Delta(x)=-\Delta(x)$. This completes the proof.
\end{proof}

\begin{rmk}
{\rm In fact, the above $\Delta_i$ ($i=1,2,3$) can be regarded as
a kind of ``1-coboundaries" that generalizes the one in Example~\ref{ex:d}. There, for $x\in A$, the derivation $$D(x):={\rm ad}(u,v)x=[u,v,x], \quad \forall u, v\in A, $$
defines a 1-cocycle on $A$.
Analogously, for $u,v,a,b\in A$, the linear map
$$\Delta_1':A\to \ot^3A, \quad x\mapsto {\rm ad}(u,v) x\otimes a \otimes b=[u,v,x]\otimes a\otimes b, \quad\forall x \in A,$$
is a 1-cocycle associated to $(A\otimes A\otimes A, {\rm
ad}\otimes 1\otimes 1 )$.
More generally, for four
families of elements $u_i,v_i,a_i,b_i\in A$, we define the 1-cocycle associated to
$(A\otimes A\otimes A, {\rm ad}\otimes 1\otimes 1 )$:
$$\Delta_1'(x)=\sum_i {\rm ad}_{u_i,v_i} x\otimes a_i \otimes b_i=
\sum_i[u_i,v_i,x]\otimes a_i\otimes b_i,\quad\forall x\in A.$$
Similarly, for $u_j',v_j',a_j',b_j', u_k'',v_k'',a_k'',
b_k''\in A$, the linear maps $$\Delta_2'(x)=\sum_j a_j'\otimes {\rm
ad}_{u_j',v_j'} x\otimes b_j',\;\; \Delta_3'(x)=\sum_k
a_k''\otimes b_k''\otimes {\rm ad}_{u_k'',v_k''} x,\;\;\forall x\in
A$$ are $1$-cocycles associated to $(A\otimes A\otimes A,1\otimes
{\rm ad}\otimes 1)$ and $(A\otimes A\otimes A,1\otimes 1\otimes
{\rm ad})$ respectively.

Moreover, set
$$\Delta=\Delta_1'+\Delta_2'+\Delta_3'.$$
From the proof of Proposition \ref{pro:skew-symmetric}, in order
for $\Delta^*$ to define a ``natural" skew-symmetric operation,
there should be some constraint conditions for the choice of the
above elements $u_i,v_i,a_i,b_i, u_j',v_j',a_j',b_j',
u_k'',v_k'',a_k'', b_k''\in A$. Here, ``natural" means that there
should not be any additional condition for the skew-symmetry of
$\Delta^*$. In particular, by a straightforward observation, it
seems reasonable to assume that the following conditions should be
satisfied: (the following sets are multi-sets, in the sense that elements can repeat in each of them)
\begin{enumerate}
\item The sets $\{u_j',v_j',a_j',b_j'\}$,
$\{u_k'',v_k'',a_k'', b_k''\}$ and
$\{u_i,v_i,a_i,b_i\}$ coincide.
\item The sets $\{u_i\}$ and $\{v_i\}$ coincide; while the sets $\{a_i\}$ and $\{b_i\}$ coincide.
\end{enumerate}
Note that the above two conditions force
$\Delta_1',\Delta_2',\Delta_3'$ to depend on only two family of
elements $\{u_i\}$ and $\{a_i\}$. With some more constraints on the indices
involving $u_i=x_i, a_i=y_i$, the $\Delta_i$ ($i=1,2,3$) given in
Eq.~(\ref{eq:delta123}) are  what we need in the above sense.

In the sequel, we will apply the notation $r=\sum_i
x_i\otimes y_i$ to represent the two families of elements $x_i,y_i\in A$ when they are used to define the cocycles $\Delta_i$ ($i=1,2,3$).
One of
the advantages of using the notation $r$ here is that $\Delta_i$
can be expressed more concisely and conventionally as
$$\Delta_1(x)=\sum_{i,j} [x_i,x_j, x]\otimes y_j\otimes
y_i=:x.[r_{13},r_{12}], \Delta_2(x)=x.[r_{21},r_{23}],
\Delta_3(x)=x.[r_{32},r_{31}],\quad \forall x\in A.$$

}
\end{rmk}

The following result is straightforward to verify.
\begin{lem}
Let $V$ be a vector space and $\Delta:V\rightarrow V\otimes
V\otimes V$  a linear map. Then $\Delta^*:V^*\otimes V^*\otimes
V^*\rightarrow V^*$ defines a $3$-Lie algebra structure on $V^*$ if and only if
$\Delta^*$ is a skew-symmetric operation and $\Delta$ satisfies
\begin{equation}\label{eq:co}
(\Delta\otimes 1\otimes 1)\Delta(x)
+\sigma_{23}\sigma_{12}((1\otimes \Delta\otimes 1)\Delta(x))+\sigma_{13}\sigma_{24}((1\otimes 1\otimes\Delta )\Delta(x))-(1\otimes 1
\otimes \Delta)\Delta(x)=0,
\end{equation}
for $x\in A$.
\end{lem}

With this preparation, we can begin our discussion on the 3-classical Yang-Baxter equation.
We introduce a notation before the next theorem.
For $a\in A$ and $1\leq i\leq 5$, define the linear map $\otimes_i a:\otimes^4 A\longrightarrow \otimes^5 A$ by inserting $a$ at the $i$-th position.
For example, for any $t=t_1\otimes t_2\otimes t_3\otimes t_4$, we have
$t\otimes_2 a= t_1\otimes a\otimes t_2\otimes t_3\otimes t_4.$

\begin{thm}\label{thm:Condition on r} Let $A$ be a $3$-Lie algebra and $r=\sum_i x_i\otimes y_i\in A\otimes A$.
Define the linear map $\Delta=\Delta_1+\Delta_2+\Delta_3:
A\rightarrow A\otimes A\otimes A$, where
$\Delta_1,\Delta_2,\Delta_3$ are given by Eq.~\eqref{eq:delta123}. Then
$\Delta^*:A^*\otimes A^*\otimes A^*\rightarrow A^*$ defines
a $3$-Lie algebra structure if and only if for any $x\in A$, the following equation holds:
\begin{eqnarray*}
&&\sum_i ({\rm ad}_{x_i, x}\otimes 1\otimes 1\otimes 1 \otimes 1) ([[r,r,r]]_1\otimes_2 y_i)+\sum_i (1\otimes{\rm ad}_{x,x_i}\otimes 1\otimes 1\otimes 1) ([[r,r,r]]_1\otimes_1 y_i)\\
&&+\sum_i (1\otimes 1\otimes{\rm ad}_{x,x_i}\otimes 1\otimes 1 ) ([[r,r,r]]_2\otimes_5 y_i)+\sum_i (1\otimes 1\otimes{\rm ad}_{x_i,x}\otimes 1\otimes 1 ) ([[r,r,r]]_2\otimes_4 y_i)\\
&&+\sum_i (1\otimes 1\otimes 1\otimes {\rm ad}_{x,x_i}\otimes 1 ) ([[r,r,r]]_2\otimes_3 y_i)+\sum_i (1\otimes 1\otimes 1\otimes {\rm ad}_{x_i,x}\otimes 1 ) ([[r,r,r]]_3\otimes_5 y_i)\\
&&+\sum_i (1\otimes 1\otimes 1\otimes 1 \otimes {\rm ad}_{x,x_i}) ([[r,r,r]]_3\otimes_4 y_i)+\sum_i (1\otimes 1\otimes 1\otimes 1\otimes{\rm ad}_{x_i,x}) ([[r,r,r]]_3\otimes_3 y_i)\\
&&=0,
\end{eqnarray*}
where
\begin{eqnarray*}~[[r,r,r]]_1&:=&[r_{12},r_{13},r_{14}]+[r_{12},r_{23},r_{24}]-[r_{13}, r_{32},r_{34}]+[r_{14},r_{42},r_{43}];\\
~[[r,r,r]]_2&:=&[r_{12},r_{31},r_{14}]-[r_{21},r_{32},r_{24}]-[r_{31}, r_{32},r_{34}]-[r_{41},r_{42},r_{34}];\\
~[[r,r,r]]_3&:=&-[r_{12},r_{13},r_{41}]+[r_{21},r_{23},r_{42}]-[r_{31},r_{32},r_{43}]-[r_{41},r_{42},r_{43}].
\end{eqnarray*}
\end{thm}
\begin{proof}  By Proposition \ref{pro:skew-symmetric}, $\Delta^*$ is skew-symmetric. Thus we only need to give the condition for which Eq.~\eqref{eq:co} holds. Since each $\Delta$ contains three terms, there are 36 terms in Eq.~(\ref{eq:co}). Let $G_i, 1\leq i\leq 5,$ denote the sum of these terms where $x$ is at the $i$-th position in the 5-tensors. We then obtain
$$
G_1+G_2+G_3+G_4+G_5=0.
$$
There are 6 terms in $G_1$:
\begin{eqnarray*}
  G_1=G_{11}+G_{12}+G_{13}+G_{14}+G_{15}+G_{16},
\end{eqnarray*}
where
{\footnotesize
$$
\begin{array}{ll}
G_{11}=\sum_{ijkl}[[x,x_i,x_j],x_k,x_l]\otimes y_l\otimes y_k\otimes y_j\otimes y_i,
    &G_{12}=\sum_{ijkl}[[x,x_i,x_j],x_k,x_l]\otimes y_l\otimes y_i\otimes y_k\otimes y_j,\\
      G_{13}=\sum_{ijkl}[[x,x_i,x_j],x_k,x_l]\otimes y_l\otimes y_j\otimes y_i\otimes y_k,
       & G_{14}=-\sum_{ijkl}[x,x_i,x_j]\otimes y_j\otimes [y_i,x_k,x_l]\otimes y_l\otimes y_k,\\
          G_{15}=-\sum_{ijkl}[x,x_i,x_j]\otimes y_j\otimes y_k\otimes [y_i,x_k,x_l]\otimes y_l,
            &G_{16}=-\sum_{ijkl}[x,x_i,x_j]\otimes y_j\otimes y_l\otimes y_k\otimes [y_i,x_k,x_l].
\end{array}
$$
}
By Condition (a) in Proposition \ref{pro:someequalities}, we have
\begin{eqnarray*}
  G_{11}+G_{12}+G_{13}&=&\sum_{ijkl}[[x_i,x_j,x_k],x_,x_l]\otimes y_l\otimes y_k\otimes y_j\otimes y_i\\
  &=&\sum_{ijkl}(\ad_{x_l,x}\otimes 1\otimes 1\otimes 1\otimes 1) [x_k,x_j,x_i]\otimes y_l\otimes y_k\otimes y_j\otimes y_i\\
  &=&\sum_l(\ad_{x_l,x}\otimes 1\otimes 1\otimes 1\otimes 1) [r_{12},r_{13},r_{14}]\otimes_2 y_l.
\end{eqnarray*}
Furthermore, we have
\begin{eqnarray*}
 G_{14}&=&\sum_{ijkl}(\ad_{x_j,x}\otimes 1\otimes 1\otimes 1\otimes 1)x_i\otimes y_j\otimes [y_i,x_l,x_k]\otimes y_l\otimes y_k\\
 &=&\sum_j(\ad_{x_j,x}\otimes 1\otimes 1\otimes 1\otimes 1)[r_{12},r_{23},r_{24}]\otimes_2y_j,
\end{eqnarray*}
and similarly,
\begin{eqnarray*}
 G_{15}&=&-\sum_j(\ad_{x_j,x}\otimes 1\otimes 1\otimes 1\otimes 1)[r_{13},r_{32},r_{34}]\otimes_2y_j,\\
 G_{16}&=&\sum_j(\ad_{x_j,x}\otimes 1\otimes 1\otimes 1\otimes 1)[r_{14},r_{42},r_{43}]\otimes_2y_j.
\end{eqnarray*}
Therefore, we obtain
$$
G_1=\sum_i ({\rm ad}_{x_i, x}\otimes 1\otimes 1\otimes 1 \otimes 1) ([[r,r,r]]_1\otimes_2 y_i).
$$
In a similar manner, we have
$$
G_2=\sum_i (1\otimes{\rm ad}_{x,x_i}\otimes 1\otimes 1\otimes 1) ([[r,r,r]]_1\otimes_1 y_i).
$$

There are 8 terms in $G_3$:
\begin{eqnarray*}
  G_3=G_{31}+G_{32}+G_{33}+G_{34}+G_{35}+G_{36}+G_{37}+G_{38},
\end{eqnarray*}
where {\footnotesize
$$
\begin{array}{ll}
 G_{31}=\sum_{ijkl}y_l\otimes y_k\otimes[[x,x_i,x_j],x_k,x_l]\otimes y_j\otimes y_i,
   & G_{32}=-\sum_{ijkl}y_j\otimes y_i\otimes[[x,x_i,x_j],x_k,x_l]\otimes y_l\otimes y_k,\\
      G_{33}=\sum_{ijkl} [y_j,x_k,x_l]\otimes y_l\otimes [x,x_i,x_j]\otimes y_k\otimes y_i,
       & G_{34}=\sum_{ijkl} y_k\otimes [y_j,x_k,x_l]\otimes[x,x_i,x_j]\otimes y_l\otimes y_i,\\
          G_{35}=\sum_{ijkl}y_l\otimes y_k\otimes[x,x_i,x_j]\otimes[y_j,x_k,x_l]\otimes y_i,
           & G_{36}=\sum_{ijkl}[y_i,x_k,x_l]\otimes y_l\otimes [x,x_i,x_j]\otimes y_j\otimes y_k,\\
            G_{37}=\sum_{ijkl}y_k\otimes[y_i,x_k,x_l]\otimes[x,x_i,x_j]\otimes y_j\otimes y_l,
            &G_{38}=\sum_{ijkl} y_l\otimes y_k\otimes[x,x_i,x_j]\otimes y_j\otimes[y_i,x_k,x_l].
\end{array}
$$
 }
We have
\begin{eqnarray*}
  G_{31}+G_{32}&=&\sum_{ijkl}y_l\otimes y_k\otimes([[x,[x_i,x_k,x_l],x_j]+[x,x_i,[x_j,x_k,x_l]])\otimes y_j\otimes y_i\\
&=&  -\sum_{ijkl}(1\otimes 1\otimes \ad_{x_j,x}\otimes 1\otimes1)y_l\otimes y_k\otimes[ x_l,x_k,x_i]\otimes y_j\otimes y_i\\
&&-\sum_{ijkl}(1\otimes 1\otimes \ad_{x,x_i}\otimes 1\otimes1)y_l\otimes y_k\otimes[ x_l,x_k,x_j]\otimes y_j\otimes y_i\\
&=&-  \sum_j(1\otimes 1\otimes \ad_{x_j,x}\otimes 1\otimes1)[r_{31},r_{32},r_{34}]\otimes_4y_j\\
&&-  \sum_i(1\otimes 1\otimes \ad_{x,x_i}\otimes 1\otimes1)[r_{31},r_{32},r_{34}]\otimes_5y_i.
\end{eqnarray*}
Furthermore, we have
\begin{eqnarray*}
  G_{33}+G_{36}&=&\sum_{ijkl}(1\otimes 1\otimes \ad_{x,x_i}\otimes 1\otimes1)[x_l,y_j,x_k]\otimes y_l\otimes x_j\otimes y_k\otimes y_i\\
  &&+\sum_{ijkl}(1\otimes 1\otimes \ad_{x_j,x}\otimes 1\otimes1)[x_l,y_i,x_k]\otimes y_l\otimes x_i\otimes y_j\otimes y_k\\
  &=&\sum_i(1\otimes 1\otimes \ad_{x,x_i}\otimes 1\otimes1)[r_{12},r_{31},r_{14}]\otimes_5 y_i\\
  &&+\sum_{j}(1\otimes 1\otimes \ad_{x_j,x}\otimes 1\otimes1)[r_{12},r_{31},r_{14}]\otimes_4 y_j,
\end{eqnarray*}
and similarly,
\begin{eqnarray*}
  G_{34}+G_{37}&=&-\sum_i(1\otimes 1\otimes \ad_{x,x_i}\otimes 1\otimes1)[r_{21},r_{32},r_{24}]\otimes_5 y_i\\
  &&-\sum_i(1\otimes 1\otimes \ad_{x_j,x}\otimes 1\otimes1)[r_{21},r_{32},r_{24}]\otimes_4 y_j,\\
  G_{35}+G_{38}&=&-\sum_{i}(1\otimes 1\otimes \ad_{x,x_i}\otimes 1\otimes1)[r_{41},r_{42},r_{34}]\otimes_5 y_i\\
  &&-\sum_{j}(1\otimes 1\otimes \ad_{x_j,x}\otimes 1\otimes1)[r_{41},r_{42},r_{34}]\otimes_4 y_j.
\end{eqnarray*}
Therefore, we obtain
$$G_3=\sum_i (1\otimes 1\otimes{\rm ad}_{x,x_i}\otimes 1\otimes 1 ) ([[r,r,r]]_2\otimes_5 y_i)+\sum_i (1\otimes 1\otimes{\rm ad}_{x_i,x}\otimes 1\otimes 1 ) ([[r,r,r]]_2\otimes_4 y_i).$$
We similarly obtain
\begin{eqnarray*}
  G_4=\sum_i \Big(1\otimes 1\otimes 1 \otimes{\rm ad}_{x,x_i}\otimes 1 ) ([[r,r,r]]_2\otimes_3 y_i)+ (1\otimes 1\otimes 1\otimes {\rm ad}_{x_i,x}\otimes 1) ([[r,r,r]]_3\otimes_5 y_i)\Big),
  \end{eqnarray*}
  \begin{eqnarray*}
  G_5=\sum_i \Big( (1\otimes 1\otimes 1\otimes 1 \otimes {\rm ad}_{x,x_i}) ([[r,r,r]]_3\otimes_4 y_i)+ (1\otimes 1\otimes 1\otimes 1 \otimes{\rm ad}_{x_i,x}) ([[r,r,r]]_3\otimes_3 y_i)\Big).
\end{eqnarray*}
This completes the proof.
\end{proof}

A direct checking gives
\begin{lem} With the notations above, if $r$ is skew-symmetric, then
$$[[r,r,r]]_1=[[r,r,r]],\;\;[[r,r,r]]_2=-[[r,r,r]],\;\;[[r,r,r]]_3=[[r,r,r]].$$
\end{lem}
We then have
\begin{cor}
Let $A$ be a $3$-Lie algebra and $r=\sum_i x_i\otimes y_i\in
A\otimes A$ skew-symmetric. Define the linear map
$\Delta=\Delta_1+\Delta_2+\Delta_3: A\rightarrow A\otimes A\otimes
A$, where $\Delta_1,\Delta_2,\Delta_3$ are induced by $r$ as in Eq.~(\ref{eq:delta123}).
Then $\Delta^*:A^*\otimes A^*\otimes A^*\rightarrow A^*$ defines a
$3$-Lie algebra structure if and only if for any $x\in A$,  the following equation
holds:
\begin{eqnarray*}
&&\sum_i ({\rm ad}_{x_i, x}\otimes 1\otimes 1\otimes 1 \otimes 1) ([[r,r,r]]\otimes_2 y_i)+\sum_i (1\otimes{\rm ad}_{x,x_i}\otimes 1\otimes 1\otimes 1) ([[r,r,r]]\otimes_1 y_i)\\
&&-\sum_i (1\otimes 1\otimes{\rm ad}_{x,x_i}\otimes 1\otimes 1 ) ([[r,r,r]]\otimes_5 y_i)-\sum_i (1\otimes 1\otimes{\rm ad}_{x_i,x}\otimes 1\otimes 1 ) ([[r,r,r]]\otimes_4 y_i)\\
&&-\sum_i (1\otimes 1\otimes 1\otimes {\rm ad}_{x,x_i}\otimes 1 ) ([[r,r,r]]\otimes_3 y_i)+\sum_i (1\otimes 1\otimes 1\otimes {\rm ad}_{x_i,x}\otimes 1 ) ([[r,r,r]]\otimes_5 y_i)\\
&&+\sum_i (1\otimes 1\otimes 1\otimes 1 \otimes {\rm ad}_{x,x_i}) ([[r,r,r]]\otimes_4 y_i)+\sum_i (1\otimes 1\otimes 1\otimes 1\otimes{\rm ad}_{x_i,x}) ([[r,r,r]]\otimes_3 y_i)\\
&&=0.
\end{eqnarray*}
In particular, if
$[[r,r,r]]=0,$
then $\Delta^*:A^*\otimes A^*\otimes A^*\rightarrow A^*$ defines
a $3$-Lie algebra structure.
\end{cor}

Summarizing the above discussions, we obtain

\begin{thm}
Let $A$ be a $3$-Lie algebra and $r\in A\otimes A$   skew-symmetric. If
$$[[r,r,r]]=0,$$
then $\Delta^*$ defines a $3$-Lie algebra structure on $A^*$, where  $\Delta=\Delta_1+\Delta_2+\Delta_3:
A\rightarrow A\otimes A\otimes A$, in which
$\Delta_1,\Delta_2,\Delta_3$ are induced by $r$ as in Eq.~\eqref{eq:delta123}. Furthermore, $(A,\Delta)$ is a \typeI.
\label{thm:ybe}
\end{thm}

Theorem~\mref{thm:ybe} can be regarded as a 3-Lie algebra analogue of the fact that a skew-symmetric solution of the classical Yang-Baxter equation gives a Lie bialgebra, leading us to give the following definition.

\begin{defi}
Let $A$ be a $3$-Lie algebra and $r\in A\otimes A$. The equation
$$[[r,r,r]]=0$$
is called the {\bf $3$-Lie classical Yang-Baxter equation (\tcybe)}.
\mlabel{defi:3cybe}
\end{defi}

This can be regarded as a natural extension of the classical Yang-Baxter equation
$$
[[r,r]]:= [r_{12},r_{13}]+[r_{12},r_{23}]+[r_{13},r_{23}]=0
$$
to the context of 3-Lie algebras.

\begin{ex}\label{ex:3d}{\rm Let $A$ be the (unique) non-trivial 3-dimensional complex 3-Lie algebra
whose non-zero product with respect to a basis $\{e_1,e_2,e_3\}$ is
given by~\cite{BSZ}
$$[e_1,e_2,e_3]=e_1.$$
If $r\in A\otimes A$ is skew-symmetric, then $r$ is a solution of the 3-Lie CYBE in $A$.
Moreover, for $r=\sum\limits_{i<j}^3r_{ij}(e_i\otimes e_j-e_j\otimes e_i)$, with the notation introduced before Theorem~\mref{thm:Condition on r}, the corresponding \typeI is given by
$$ \Delta_i(e_1)=(-1)^i r_{23}r\ot_i e_1, \quad
\Delta_i(e_2)=(-1)^{i+1}r_{13}r\ot_i e_1, \quad
\Delta_i(e_3)=(-1)^{i}r_{12}r\ot_i e_1,$$
and the comultiplication $\Delta:A\longrightarrow\wedge^3A$ is given by
$$\Delta(e_1)=-r_{23}^2e_1\wedge e_2\wedge e_3,\;\;\Delta(e_2)=r_{13}r_{23}e_1\wedge e_2\wedge e_3,\;\;
\Delta(e_3)=-r_{12}r_{23}e_1\wedge e_2\wedge e_3,$$ where $e_1\wedge
e_2\wedge e_3=\sum_{\sigma \in S_3}{\rm sgn}(\sigma)
e_{\sigma(1)}\otimes e_{\sigma(2)}\otimes e_{\sigma(3)}$ and $S_3$
is the permutation group on $\{1,2,3\}$. In particular, when $r_{23}\ne
0$, we get a \typeI whose coproduct is not zero. }
\end{ex}

Let $r\in A\otimes A$. Then $r$ induces a linear map $A^*\rightarrow A$ that we still denote by $r$:
\begin{equation}
\langle r(\xi), \eta\rangle=\langle r, \xi\otimes \eta\rangle,\quad \forall \xi, \eta\in A^*.
\label{eq:opform}
\end{equation}
Furthermore, we denote the ternary operation $\Delta^*:A^*\otimes A^*\otimes A^*\rightarrow A^*$ by $[\cdot,\cdot,\cdot]^*$.

\begin{pro}
Let $A$ be a $3$-Lie algebra and $r\in A\otimes A$. Suppose that
$r$ is skew-symmetric and $\Delta=\Delta_1+\Delta_2+\Delta_3:
A\rightarrow A\otimes A\otimes A$, in which
$\Delta_1,\Delta_2,\Delta_3$ are induced by $r$ as in
Eq.~\eqref{eq:delta123}. Then we have
\begin{equation}\label{eq:dual bracket}
  [\xi,\eta,\gamma]^*={\rm ad}^*_{r(\xi),r(\eta)}\gamma+{\rm ad}^*_{r(\eta),r(\gamma)}\xi+{\rm ad}^*_{r(\gamma),r(\xi)}\eta,\quad
  \forall \xi,\eta,\gamma\in A^*.
\end{equation}
Furthermore, we have
\begin{equation}\label{eq:rformula}
[r(\xi),r(\eta),r(\gamma)]- r([\xi,\eta,\gamma]^*)=[[r,r,r]](\xi,\eta,\gamma),\quad
  \forall \xi,\eta,\gamma\in A^*.
\end{equation}
\end{pro}

\begin{proof}
Let $x\in A, \xi,\eta,\gamma\in A^*$. For the first conclusion, we only need to prove
$$\langle \Delta(x),\xi\otimes \eta\otimes \gamma\rangle=\langle x,[\xi,\eta,\gamma]^*\rangle.$$Let $r=\sum_ix_i\otimes y_i$. Since $r$ is skew-symmetric, we have
\begin{eqnarray*}
\langle x, {\rm ad}^*_{r(\xi),r(\eta)}\gamma\rangle
&=& \langle -[r(\xi),r(\eta),x],\gamma\rangle
=-\langle r,\eta\otimes {\rm ad}^*_{r(\xi),x}\gamma\rangle\\
&=&\sum_i\langle y_i, \eta\rangle \langle x_i, {\rm ad}^*_{r(\xi),x}\gamma\rangle
=\sum_i\langle y_i, \eta\rangle \langle r,\xi\otimes {\rm ad}^*_{x, x_i}\gamma\rangle\\
&=&\sum_{i,j}\langle y_i, \eta\rangle \langle y_j,\xi\rangle\langle [x,x_i,x_j],\gamma\rangle
=\Big\langle \sum_{ij} y_j\otimes y_i\otimes [x,x_i,x_j], \xi\otimes \eta\otimes \gamma\Big\rangle\\
&=&\langle\Delta_3(x),\xi\otimes \eta\otimes \gamma\rangle.
\end{eqnarray*}
Similarly, we have
$$  \langle x, {\rm ad}^*_{r(\eta),r(\gamma)}\xi\rangle
=\langle\Delta_1(x),\xi\otimes \eta\otimes \gamma\rangle,\quad
\langle x, {\rm ad}^*_{r(\gamma),r(\xi)}\eta\rangle=\langle\Delta_2(x),\xi\otimes \eta\otimes \gamma\rangle.
$$
Therefore, we obtain
\begin{eqnarray*}
  \langle \Delta(x),\xi\otimes \eta\otimes \gamma\rangle&=&\langle \Delta_1(x)+\Delta_2(x)+\Delta_3(x),\xi\otimes \eta\otimes \gamma\rangle\\
  &=&\langle x, {\rm ad}^*_{r(\eta),r(\gamma)}\xi\rangle+\langle x, {\rm ad}^*_{r(\gamma),r(\xi)}\eta\rangle+\langle x, {\rm ad}^*_{r(\xi),r(\eta)}\gamma\rangle\\
 &=&\langle x,[\xi,\eta,\gamma]^*\rangle.
\end{eqnarray*}
This finishes the proof of Eq.~\eqref{eq:dual bracket}.
\smallskip

Applying the left hand side of Eq.~\eqref{eq:rformula} to $\kappa\in A^*$ gives
\begin{eqnarray*}
\langle [r(\xi),r(\eta),r(\gamma)]-r([\xi,\eta,\gamma]^*),\kappa\rangle
&=&\langle\kappa,[r(\xi),r(\eta),r(\gamma)]\rangle-\langle\gamma,[r(\xi),r(\eta),r(\kappa)]\rangle\\
&&-\langle\xi,[r(\eta),r(\gamma),r(\kappa)]\rangle-\langle\eta,[r(\gamma),r(\xi),r(\kappa)]\rangle.
\end{eqnarray*}
Applying the right hand side of Eq.~\eqref{eq:rformula} to $\kappa\in A^*$ gives
\allowdisplaybreaks
{\begin{eqnarray*}
 &&[[r,r,r]](\xi,\eta,\gamma,\kappa)\\&=&\sum_{i,j,k}\Big([x_i,x_j,x_k]\otimes y_i\otimes y_j\otimes y_k(\xi,\eta,\gamma,\kappa)
 +x_i\otimes [y_i,x_j,x_k]\otimes y_j\otimes y_k(\xi,\eta,\gamma,\kappa)\\
  &&+x_i\otimes x_j\otimes [y_i,y_j,x_k]\otimes y_k(\xi,\eta,\gamma,\kappa)
  +x_i\otimes x_j\otimes x_k\otimes [y_i,y_j,y_k](\xi,\eta,\gamma,\kappa)\Big)\\
  &=&\sum_{i,j,k}\Big(\langle\xi,[x_i,x_j,x_k]\rangle \langle \eta, y_i\rangle\langle\gamma, y_j\rangle\langle\kappa, y_k\rangle+\langle\eta,[y_i,x_j,x_k]\rangle \langle \xi, x_i\rangle\langle\gamma, y_j\rangle\langle\kappa, y_k\rangle\\
 &&+\langle\gamma,[y_i,y_j,x_k]\rangle \langle \xi, x_i\rangle\langle\eta, x_j\rangle\langle\kappa, y_k\rangle+\langle\kappa,[y_i,y_j,y_k]\rangle \langle \xi, x_i\rangle\langle\eta, x_j\rangle\langle\gamma,x_k\rangle\Big)\\
 &=&-\langle\xi,[r(\eta),r(\gamma),r(\kappa)]\rangle-\langle\eta,[r(\gamma),r(\xi),r(\kappa)]\rangle\\
 &&-\langle\gamma,[r(\xi),r(\eta),r(\kappa)]\rangle+\langle\kappa,[r(\xi),r(\eta),r(\gamma)]\rangle.
\end{eqnarray*}
}
Therefore, Eq.~\eqref{eq:rformula} holds. This completes the proof.
\end{proof}
As a direct consequence, we obtain
\begin{cor}
Suppose that $r$ is a skew-symmetric solution of the $3$-Lie classical Yang-Baxter equation. Then the linear map $r$ defined in Eq.~(\ref{eq:opform}) is a $3$-Lie algebra morphism from $(A^*,[\cdot,\cdot,\cdot]^*)$ to $(A,[\cdot,\cdot,\cdot])$.
\end{cor}

We further give the following
interpretation of the invertible skew-symmetric solutions of the
$3$-Lie classical Yang-Baxter equation which is parallel to a
similar result for the classical Yang-Baxter equation in a Lie
algebra given by Drinfeld \cite{D}.

\begin{pro}
Let $A$ be a $3$-Lie algebra and $r\in A\otimes A$. Suppose that $r$ is skew-symmetric and
nondegenerate. Then $r$ is a solution of the $3$-Lie classical Yang-Baxter equation if and only if
the nondegenerate skew-symmetric bilinear form $B$ on $A$ defined by $B(x,y):=\langle r^{-1}(x), y\rangle$ for any $x,y\in A$ satisfies
\begin{equation}\label{eq:3sb}
B([x,y,z],w)-B([x,y,w],z)+B([x,z,w],y)-B([y,z,w],x)=0,\quad \forall x,y,z,w\in A.
\end{equation}
\end{pro}

\begin{proof}
For any $x,y,z,w\in A$, since $r$ is nondegenerate, there exist $\xi,\eta,\gamma,\kappa\in A^*$ such that $r(\xi)=x, r(\eta)=y,r(\gamma)=z, r(\kappa)=w$.
By Eq.~\eqref{eq:rformula}, if $[[r,r,r]]=0,$ we have
\begin{eqnarray*}
B([x,y,z],w)&=&\langle r^{-1}[r(\xi),r(\eta),r(\gamma)], w\rangle\\
&=&\langle {\rm ad}^*_{r(\xi),r(\eta)}\gamma+{\rm ad}^*_{r(\eta),r(\gamma)}\xi+{\rm ad}^*_{r(\gamma),r(\xi)}\eta, w\rangle\\
&=&\langle -\gamma, [x,y,w]\rangle -\langle \xi, [y,z,w]\rangle -\langle \eta, [z,x,w]\rangle\\
&=&-B(z, [x,y,w])-B(x, [y,z,w])-B(y, [z,x,w]).
\end{eqnarray*}
Hence the conclusion follows.
\end{proof}

\subsection{$\mathcal O$-operators, 3-pre-Lie algebras and solutions of the 3-Lie CYBE}
\mlabel{ss:oop}
Pre-Lie algebras and $\mathcal O$-operators are useful tools for the construction of solutions of the classical Yang-Baxter equation.  For further details, see~\cite{leftsymm4} for pre-Lie algebras and~\cite{Bai-Unifiedapproach} for the relationship between pre-Lie algebras, $\mathcal O$-operators and CYBE.
In this subsection, we first introduce the notion of an $\mathcal
O$-operator in the 3-Lie algebra context, which could give solutions of the 3-Lie classical
Yang-Baxter equation. Then we introduce the notion of a 3-pre-Lie
algebra which is closely related to $\mathcal O$-operators. In particular,
there is a construction of solutions of the 3-Lie classical
Yang-Baxter equation in some special 3-Lie algebras obtained from
3-pre-Lie algebras.

\begin{defi}\label{defi:o}
Let $(A,[\cdot,\cdot,\cdot])$ be a $3$-Lie algebra and $(V,\rho)$
a representation.  A linear operator $T:V\rightarrow A$ is called
an {\bf $\mathcal O$-operator} associated to $( V,\rho)$ if $T$
satisfies
\begin{equation}\label{eq:Ooperator}
 [Tu,Tv,Tw]=T\left(\rho(Tu,Tv)w+\rho(Tv,Tw)u+\rho(Tw,Tu)v\right),\quad \forall u,v,w\in V.
\end{equation}
\end{defi}

\begin{ex}
{\rm Let $A$ be a 3-Lie algebra and $r\in A\otimes A$. Suppose that $r$ is skew-symmetric. Then $r$
is a solution of the 3-Lie classical Yang-Baxter equation if and only if $r$ is an $\mathcal O$-operator of $A$ associated
to the coadjoint representation $(A^*,{\rm ad}^*)$.}
\end{ex}

\begin{ex}
{\rm Let $A$ be a 3-Lie algebra. An $\mathcal O$-operator of $A$ associated
to the adjoint representation $(A,{\rm ad})$ is called a {\bf Rota-Baxter operator of weight zero}. See \cite{BaiRGuo} for more details.}
\end{ex}

 Let  $T:V\rightarrow A$ be a linear map and $\overline{T}\in V^*\otimes A$ the corresponding tensor, i.e.
$$
\overline{T}(v,\xi)=\langle\xi,Tv\rangle,\quad \forall~v\in V,~\xi\in A^*.
$$

The following result is the 3-Lie algebra analogue of the relationship between $\mathcal O$-operators and the classical Yang-Baxter equation on  Lie algebras~\cite{Bai-Unifiedapproach}.
\begin{thm}\label{thm:O-r}
With the above notations,
$T$ is an $\mathcal O$-operator if and only if
\begin{equation}
r=\overline{T}-\sigma_{12}(\overline{T})
\end{equation}
is a skew-symmetric solution of the $3$-Lie classical Yang-Baxter equation in the semi-direct product $3$-Lie algebra $A\ltimes_{\rho^*}V^*$.
\end{thm}

\pf
Let $\{v_1,\cdots,v_n\}$ be a basis of $V$ and $\{v_1^*,\cdots,v_n^*\}$ the dual basis. Then we have
$$\overline{T}=\sum_i v_i^*\otimes Tv_i\in V^*\otimes A.$$
Therefore, we derive
\begin{eqnarray*}
~[r_{12},r_{13},r_{14}]&=& \sum_{ijk}\Big(-[Tv_i,Tv_j,Tv_k]\otimes v_i^*\otimes v_j^*\otimes v_k^*
+[Tv_i,v_j^*,Tv_k]\otimes v_i^*\otimes Tv_j\otimes v_k^*\\
&\mbox{}&+[Tv_i,Tv_j,v_k^*]\otimes v_i^*\otimes v_j^*\otimes Tv_k+
[v_i^*,Tv_j,Tv_k]\otimes Tv_i\otimes v_j^*\otimes v_k^*\Big);\\
~[r_{12},r_{23},r_{24}]&=& \sum_{ijk}\Big(v_i^*\otimes [Tv_i,Tv_j,Tv_k]\otimes v_j^*\otimes v_k^*-
Tv_i\otimes [v_i^*,Tv_j,Tv_k]\otimes v_j^*\otimes v_k^*\\
&\mbox{}&- v_i^*\otimes [Tv_i,Tv_j,v_k^*]\otimes v_j^*\otimes v_k^*-
v_i^*\otimes [Tv_i,v_j^*,Tv_k]\otimes Tv_j\otimes v_k^*\Big);\\
~[r_{13},r_{23},r_{34}]&=& \sum_{ijk}\Big(-
v_i^*\otimes v_j^*\otimes [Tv_i,Tv_j,Tv_k]\otimes v_k^*+Tv_i\otimes v_j^*\otimes[v_i^*,Tv_j,Tv_k]\otimes v_k^*\\
&\mbox{}& +v_i^*\otimes Tv_j\otimes [Tv_i,v_j^*,Tv_k]\otimes v_k^*+v_i^*\otimes v_j^*\otimes [Tv_i,Tv_j,v_k^*]\otimes Tv_k\Big);\\
~[r_{14},r_{24},r_{34}]&=& \sum_{ijk}\Big(
-Tv_i\otimes v_j^*\otimes v_k^*\otimes [v_i^*,Tv_j,Tv_k]
+v_i^*\otimes v_j^*\otimes v_k^*\otimes [Tv_i,Tv_j,Tv_k]\\
&\mbox{}& -v_i^*\otimes Tv_j\otimes v_k^*\otimes [Tv_i,v_j^*,Tv_k]-v_i^*\otimes v_j^*\otimes Tv_k\otimes [Tv_i,Tv_j,v_k^*]\Big).
\end{eqnarray*}
Set
$$OT(u,v,w)=[Tu,Tv,Tw]-T \Big(\rho(Tu,Tv)w+\rho(Tv,Tw)u+\rho(Tw,Tu)v\Big),\quad\forall u,v,w\in A.$$
Note
\begin{eqnarray*}
\sum_iTv_i\otimes [v_i^*,Tv_j,Tv_k]&=&\sum_i Tv_i\otimes \rho^*(Tv_j,Tv_k)v_i^*=\sum_i Tv_i\otimes \sum_m  \langle \rho^*(Tv_j,Tv_k)v_i^*,v_m\rangle v_m^*\\&=&\sum_{i,m}Tv_i\otimes -\langle \rho(Tv_j,Tv_k)v_m, v_i^*\rangle v_m^*=\sum_m -T(\rho(Tv_j,Tv_k)v_m\otimes v_m^*.
\end{eqnarray*}
Thus, we have
\begin{eqnarray*}
[[r,r,r]]&=& [r_{12},r_{13},r_{14}]+[r_{12},r_{23},r_{24}] +[r_{13},r_{23},r_{34}] +[r_{14},r_{24},r_{34}]\\
&=&\sum_{i,j,k}\Big( -OT(v_i,v_j,v_k)\otimes v_i^*\otimes v_j^*\otimes v_k^*
 +v_i^*\otimes OT(v_i,v_j,v_k)\otimes v_j^*\otimes v_k^* \\
 &&-v_i^*\otimes v_j^*\otimes OT(v_i,v_j,v_k)\otimes v_k^* +v_i^*\otimes v_j^*\otimes v_k^*\otimes OT(v_i,v_j,v_k)\Big).
\end{eqnarray*}
Therefore $r$ satisfies the 3-Lie classical Yang-Baxter equation, i.e. $[[r,r,r]]=0$ if and only if $OT(v_i,v_j,v_k)=0$ for all $i,j,k$, i.e.
$T$ is an $\mathcal O$-operator.
\qed\vspace{3mm}

\begin{defi}
Let $A$ be a vector space with a linear map $\{\cdot,\cdot,\cdot\}:A\otimes A\otimes A\rightarrow A$.
The pair $(A,\{\cdot,\cdot,\cdot\})$ is called a {\bf $3$-pre-Lie algebra} if the following identities hold:
\begin{eqnarray}
\{x,y,z\}&=&-\{y,x,z\},\label{eq:d1}\\
\nonumber\{x_1,x_2,\{x_3,x_4,x_5\}\}&=&\{[x_1,x_2,x_3]_C,x_4,x_5\}+\{x_3,[x_1,x_2,x_4]_C,x_5\}\\
&&+\{x_3,x_4,\{x_1,x_2,x_5\}\},\label{eq:d2}\\
\nonumber\{ [x_1,x_2,x_3]_C,x_4, x_5\}&=&\{x_1,x_2,\{ x_3,x_4, x_5\}\}+\{x_2,x_3,\{ x_1,x_4,x_5\}\}\\
&&+\{x_3,x_1,\{ x_2,x_4, x_5\}\},\label{eq:d3}
\end{eqnarray}
 where $x,y,z, x_i\in A, 1\leq i\leq 5$ and $[\cdot,\cdot,\cdot]_C$ is defined by
\begin{equation}
[x,y,z]_C=\{x,y,z\}+\{y,z,x\}+\{z,x,y\},\quad \forall  x,y,z\in A.\label{eq:3cc}
\end{equation}
\end{defi}
This agrees with the general construction of splitting of operads applied to the operad of the $3$-Lie algebra~\mcite{PBG}.

\begin{pro}
Let $(A,\{\cdot,\cdot,\cdot\})$ be a $3$-pre-Lie algebra. Then the induced $3$-commutator given by Eq.~\eqref{eq:3cc} defines
a $3$-Lie algebra.
\end{pro}

\begin{proof} By Eq.~\eqref{eq:d1}, the induced 3-commutator $[\cdot,\cdot,\cdot]_C$ given by Eq.~\eqref{eq:3cc} is skew-symmetric. For $x_1,x_2,x_3,x_4,x_5\in A$, we have
\begin{eqnarray*}
&&[x_1,x_2,[x_3,x_4,x_5]_C]_C-[[x_1,x_2,x_3]_C,x_4,x_5]_C-[x_3,[x_1,x_2,x_4]_C,x_5]_C\\
&&-[x_3,x_4,[ x_1,x_2, x_5]_C]_C\\
&=&\{x_1,x_2,\{x_3,x_4,x_5\}\}+\{x_1,x_2,\{x_4,x_5,x_3\}\}+\{x_1,x_2,\{x_5,x_3,x_4\}\}\\
&&+\{x_2,[x_3,x_4,x_5]_C, x_1\}+\{[x_3,x_4,x_5]_C, x_1, x_2\}\\
&&-\{x_4,x_5,\{x_1,x_2,x_3\}\}-\{x_4,x_5,\{x_2,x_3,x_1\}\}-\{x_4,x_5,\{x_3,x_1,x_2\}\}\\
&&-\{[x_1,x_2,x_3]_C,x_4,x_5\}-\{x_5,[x_1,x_2,x_3]_C,x_4\}\\
&&-\{x_5,x_3,\{x_1,x_2,x_4\}\}-\{x_5,x_3,\{x_2,x_4,x_1\}\}-\{x_5,x_3,\{x_4,x_1,x_2\}\}\\
&&-\{x_3,[x_1,x_2,x_4]_C,x_5\}-\{[x_1,x_2,x_4]_C,x_5,x_3\}\\
&&-\{x_3,x_4,\{ x_1,x_2, x_5\}\}-\{x_3,x_4,\{ x_2,x_5, x_1\}\}-\{x_3,x_4,\{ x_5,x_1, x_2\}\}\\
&&-\{x_4,[ x_1,x_2, x_5]_C,x_3\}-\{[ x_1,x_2, x_5]_C,x_3,x_4\}\\
&=&0.
\end{eqnarray*}
This is because
\begin{eqnarray*}
  \{x_1,x_2,\{x_3,x_4,x_5\}\}&=&\{[x_1,x_2,x_3]_C,x_4,x_5\}+\{x_3,[x_1,x_2,x_4]_C,x_5\} +\{x_3,x_4,\{ x_1,x_2, x_5\}\},\\
  \{x_1,x_2,\{x_4,x_5,x_3\}\}&=&\{[x_1,x_2,x_4]_C,x_5,x_3\}+\{x_4,[ x_1,x_2, x_5]_C,x_3\} +\{x_4,x_5,\{x_1,x_2,x_3\}\},\\
  \{x_1,x_2,\{x_5,x_3,x_4\}\}&=&\{x_5,[x_1,x_2,x_3]_C,x_4\}+\{[ x_1,x_2, x_5]_C,x_3,x_4\} +\{x_5,x_3,\{x_1,x_2,x_4\}\},\\
  \{x_2,[x_3,x_4,x_5]_C, x_1\}&=&\{x_4,x_5,\{x_2,x_3,x_1\}\}+\{x_5,x_3,\{x_2,x_4,x_1\}\}+\{x_3,x_4,\{ x_2,x_5, x_1\}\},\\
  \{[x_3,x_4,x_5]_C, x_1, x_2\}&=&\{x_3,x_4,\{ x_5,x_1, x_2\}\} +\{x_4,x_5,\{x_3,x_1,x_2\}\}+\{x_5,x_3,\{x_4,x_1,x_2\}\}.
\end{eqnarray*}
Thus the proof is completed.
\end{proof}

\begin{defi}
Let $(A,\{\cdot,\cdot,\cdot\})$ be a $3$-pre-Lie algebra. The $3$-Lie algebra $(A,[\cdot,\cdot,\cdot]_C)$
is called the {\bf sub-adjacent $3$-Lie algebra} of $(A,\{\cdot,\cdot,\cdot\})$ and $(A,\{\cdot,\cdot,\cdot\})$ is called a {\bf compatible
$3$-pre-Lie algebra} of the $3$-Lie algebra $(A,[\cdot,\cdot,\cdot]_C)$.
\end{defi}

Let $(A,\{\cdot,\cdot,\cdot\})$ be a $3$-pre-Lie algebra. Define a skew-symmetric  linear map $L: \otimes^2A\rightarrow \gl(A)$
by
\begin{equation}\label{eq:R}
L(x,y)z=\{x,y,z\},\quad \forall x,y,z\in A.
\end{equation}

By the definitions of a $3$-pre-Lie algebra and a representation of a $3$-Lie algebra, we immediately obtain

\begin{pro}
With the above notations, $(A,L)$ is a representation of the
$3$-Lie algebra $(A,[\cdot,\cdot,\cdot]_C)$. On the other hand,
let $A$ be a vector space with a linear map
$\{\cdot,\cdot,\cdot\}:A\otimes A\otimes A\rightarrow A$
satisfying
 Eq.~\eqref{eq:d1}. Then $(A,\{\cdot,\cdot,\cdot\}) $ is a $3$-pre-Lie algebra if $[\cdot,\cdot,\cdot]_C$ defined by Eq~\eqref{eq:3cc} is a $3$-Lie algebra and the left multiplication $L$ defined by Eq.~\eqref{eq:R}
gives a representation of this $3$-Lie algebra.
\end{pro}

New identities of 3-pre-Lie algebras can be derived from Proposition~\ref{pro:new}. For example,

\begin{cor}
Let $(A,\{\cdot,\cdot,\cdot\})$ be a $3$-pre-Lie algebra. Then the following identities hold:
\begin{eqnarray*}
\{ [x_1,x_2,x_3]_C,x_4, x_5\}-\{[x_1,x_2,x_4]_C,x_3, x_5 \}\\
+\{[x_1,x_3,x_4]_C,x_2, x_5\}
-\{[x_2,x_3,x_4]_C,x_1, x_5\}&=&0,\\
\{x_1,x_2,\{x_3,x_4, x_5\}\}+\{x_3,x_4,\{x_1,x_2, x_5\}\}+\{x_2,x_4,\{x_3,x_1, x_5\}\}\\\nonumber
+\{x_3,x_1,\{x_2,x_4, x_5\}\}+\{x_2,x_3,\{x_1,x_4, x_5\}\}+\{x_1,x_4,\{x_2,x_3, x_5\}\}&=&0,
\end{eqnarray*}
for $x_i\in A, 1\leq i\leq 5$.
\end{cor}

\begin{pro}\label{pro:3preLieT}
Let $(A,[\cdot,\cdot,\cdot])$ be a $3$-Lie algebra and $(V,\rho)$ a representation. Suppose that the linear map $T:V\rightarrow A$ is an $\mathcal O$-operator associated
to $(V,\rho)$. Then there exists a $3$-pre-Lie algebra structure on $V$ given by
\begin{equation}
\{u,v,w\}=\rho(Tu,Tv)w,\quad\forall ~ u,v,w\in V.
\end{equation}
\end{pro}

\begin{proof} Let $u,v,w\in V$.
It is obvious that
$$\{u,v,w\}=\rho(Tu,Tv)w=-\rho(Tv,Tu)w=-\{v,u,w\}.$$
Furthermore, the following equation holds.
$$
[u,v,w]_C=\rho(Tu,Tv)w+\rho(Tv,Tw)u+\rho(Tw,Tu)v.
$$
Since $T$ is an $\mathcal O$-operator, we have
$$
T[u,v,w]_C=[Tu,Tv,Tw].
$$
For $v_1,v_2,v_3,v_4,v_5\in V$, we have
\begin{eqnarray*}
\{v_1,v_2,\{v_3,v_4,v_5 \}\}&=&\rho(Tv_1,Tv_2)\rho(Tv_3,Tv_4)v_5;\\
\{[v_1,v_2,v_3]_C,v_4,v_5\}&=&\rho(T[v_1,v_2,v_3]_C,Tv_4)v_5=\rho([Tv_1,Tv_2,Tv_3],Tv_4)v_5;\\
\{v_3,[v_1,v_2,v_4]_C,v_5\}&=&\rho(Tv_3,T[v_1,v_2,v_4]_C)v_5=\rho(Tv_3,[T v_1,Tv_2,Tv_4])v_5;\\
\{v_3,v_4,\{v_1,v_2,v_5\}\}&=&\rho(Tv_3,Tv_4)\rho(Tv_1,Tv_2)v_5.
\end{eqnarray*}
Since $(V,\rho)$ is a representation of $A$, by Condition (i) in Definition \ref{defi:rep}, Eq.~(\ref{eq:d2}) holds. On the other hand, we have
\begin{eqnarray*}
\{[v_1,v_2,v_3]_C,v_4, v_5\}&=&\rho(T[v_1,v_2,v_3]_C,Tv_4)v_5=\rho([Tv_1,Tv_2,Tv_3],Tv_4)v_5;\\
\{v_1,v_2,\{v_3,v_4, v_5\}\}&=&\rho(Tv_1,Tv_2)\rho(Tv_3,Tv_4)v_5;\\
\{v_2,v_3,\{v_1,v_4, v_5\}\}&=&\rho(Tv_2,Tv_3)\rho(Tv_1,Tv_4)v_5;\\
\{v_3,v_1,\{v_2,v_4, v_5\}\}&=&\rho(Tv_3,Tv_1)\rho(Tv_2,Tv_4)v_5.
\end{eqnarray*}
By Condition (ii) in Definition \ref{defi:rep}, Eq.~(\ref{eq:d3})
holds. This completes the proof.
\end{proof}

\begin{cor}
With the above conditions,  $(V,[\cdot,\cdot,\cdot]_C)$ is a $3$-Lie
algebra as the sub-adjacent $3$-Lie algebra of the $3$-pre-Lie
algebra given in Proposition \ref{pro:3preLieT}, and $T$ is a $3$-Lie algebra morphism from $(V,[\cdot,\cdot,\cdot]_C)$ to $(A,[\cdot,\cdot,\cdot])$. Furthermore,
$T(V)=\{Tv|v\in V\}\subset A$ is a $3$-Lie subalgebra of $A$ and there is an induced $3$-pre-Lie algebra structure $\{\cdot,\cdot,\cdot\}_{T(V)}$ on
$T(V)$ given by
\begin{equation}
\{Tu,Tv,Tw\}_{T(V)}:=T\{u,v,w\},\quad\;\forall u,v,w\in V.
\end{equation}
\end{cor}

\begin{pro}\label{pro:preLieOoper}
Let $(A,[\cdot,\cdot,\cdot])$ be a $3$-Lie algebra. Then there exists a compatible $3$-pre-Lie algebra if and only if there exists an invertible $\mathcal O$-operator on $A$.
\end{pro}

\begin{proof}
Let $T$ be an invertible $\mathcal O$-operator of $A$ associated to a representation $(V,\rho)$. Then there exists a 3-pre-Lie algebra structure on $V$
defined by
$$\{u,v,w\}=\rho(Tv,Tw)u,\quad \forall u,v,w\in V.$$
Moreover, there is an induced 3-pre-Lie algebra structure $\{\cdot,\cdot,\cdot\}_A$ on $A=T(V)$ given by
$$\{x,y,z\}_A=T\{T^{-1}x,T^{-1}y,T^{-1}z\}=T\rho(x,y)T^{-1}(z)$$
for all $x,y,z\in A$. Since $T$ is an $\mathcal O$-operator, we have
\begin{eqnarray*}
[x,y,z]&=&T\left(\rho (y,z)T^{-1}(x)+\rho (z,x)T^{-1}(y)+\rho (x,y)T^{-1}z\right)\\
&=&\{x,y,z\}_A+\{y,z,x\}_A+\{z,x,y\}_A.\end{eqnarray*}
 Therefore $(A,\{\cdot,\cdot,\cdot\}_A)$ is a compatible 3-pre-Lie algebra.
Conversely, the identity map ${\rm id}$ is an $\mathcal O$-operator of the sub-adjacent 3-Lie algebra of a 3-pre-Lie algebra
associated to the representation $(A, L)$.
\end{proof}

\begin{pro}
Let $(A,[\cdot,\cdot,\cdot])$ be a $3$-Lie algebra, and $B$  a nondegenerate skew-symmetric bilinear form satisfying Eq.~\eqref{eq:3sb}. Then there exists
a compatible $3$-pre-Lie algebra structure on $A$ given by
\begin{equation}
B(\{x,y,z\},w)=-B(z,[x,y,w]),\quad \forall  x,y,z,w\in A.
\end{equation}
\end{pro}

\begin{proof}
Define a linear map $T:A^*\longrightarrow A$  by $\langle T^{-1}x,y\rangle=B(x,y)$, or equivalently, $B(T\xi,y)=\langle\xi,y\rangle$ for all $x,y\in A$ and $\xi\in A^*$.  By Eq.~\eqref{eq:3sb}, we obtain that $T$ is an invertible $\mathcal O$-operator associated to
the coadjoint representation $(A^*,{\rm ad}^*)$. By Proposition \ref{pro:preLieOoper},   there exists a compatible 3-pre-Lie algebra on $A$ given by $\{x,y,z\}_A=T({\rm ad}^*_{x,y}T^{-1}z)$ for $x,y,z\in A$. Then we have
\begin{eqnarray*}
B(\{x,y,z\}_A,w)&=&B(T({\rm ad}^*_{x,y}T^{-1}z),w)=\langle {\rm ad}^*_{x,y}T^{-1}z, w\rangle\\
&=&\langle T^{-1}(z), -[x,y,w]\rangle=-B(z,[x,y,w]),
\end{eqnarray*}
for $x,y,z,w\in A$. This completes the proof.
\end{proof}

We end this subsection by obtaining solutions of  the $3$-Lie classical Yang-Baxter equation from  $3$-pre-Lie algebras.
\begin{thm} Let $(A,[\cdot,\cdot,\cdot])$ be a $3$-pre-Lie algebra. Let $\{e_i\}$ be a basis of $A$ and $\{e_i^*\}$   the dual basis. Then
\begin{equation}
r=\sum_{i}e_i\otimes e_i^*-e_i^*\otimes e_i\
\end{equation}
is a skew-symmetric solution of the $3$-Lie classical Yang-Baxter equation in $A\ltimes_{L^*}A^*$.
\end{thm}

\begin{proof}
It follows from Theorem~\ref{thm:O-r} and the fact that the identity map ${\rm id}$ is an $\mathcal O$-operator of the sub-adjacent 3-Lie algebra of a 3-pre-Lie algebra
associated to the representation $( A, L)$.
\end{proof}

\section{Double construction 3-Lie bialgebras}
\label{sec:triple}

We consider \typeIIs in this section. In
Section 4.1, we introduce the notion of a Manin triple and a matched pair of 3-Lie
algebras. In Section 4.2, we introduce the notion
of a \typeII, and establish the correspondence
between a Manin triple of 3-Lie algebras, a matched pair of 3-Lie
algebras and a \typeII (Theorem
\ref{thm:relations}). In Section 4.3, we establish the relationship between \typeIs and \typeIIs, and provide explicit examples of the
latter.

\subsection{Manin triples and matched pairs of 3-Lie algebras}

\begin{defi}
Let $A$ be a $3$-Lie algebra. A bilinear form $(\cdot,\cdot)_A$  on $A$ is called {\bf invariant} if it satisfies
\begin{equation}
 ([x_1,x_2,x_3],x_4)_A+([x_1,x_2,x_4],x_3)_A=0,\quad\forall  x_1,x_2,x_3,x_4\in A.
\end{equation}
A $3$-Lie algebra $A$ is called {\bf pseudo-metric} if there is a nondegenerate symmetric invariant bilinear form on $A$.
\end{defi}

\begin{defi}\label{defi:Manin}
A {\bf Manin triple of $3$-Lie algebras} consists of a pseudo-metric $3$-Lie algebra $(\huaA,(\cdot,\cdot)_\huaA)$ and $3$-Lie algebras $A_1, A_2$   such that
\begin{enumerate}
\item
$A_1$ and $A_2$ are $3$-Lie subalgebras of $\huaA$ which are isotropic, that is,
$(x_1,y_1)_\huaA=(x_2,y_2)_\huaA=0$ for $x_1,y_1\in A_1$ and
$x_2,y_2\in A_2$;
\item $\huaA=A_1\oplus A_2$ as the direct sum of vector spaces;
\item For all $x_1,y_1\in A_1 $ and $x_2,y_2\in A_2$, we have $\pr_1[x_1,y_1,x_2]=0$ and $\pr_2[x_2,y_2,x_1]=0$, where $\pr_1$ and $\pr_2$ are the projections from $A_1\oplus A_2$ to $A_1$ and $A_2$ respectively.
\end{enumerate}
An {\bf isomorphism} between two Manin triples $((\huaA,(\cdot,\cdot)_\huaA),A_1,A_2)$ and $((\huaA',(\cdot,\cdot)_{\huaA'}),A'_1,A'_2)$ is an isomorphism $f:\huaA\rightarrow \huaA'$ of
$3$-Lie algebras satisfying
$$f(A_1)\subset A'_1, \quad f(A_2)\subset A'_2, \quad  (x_1,x_2)_\huaA=(f(x_1),f(x_2))_{\huaA'},\quad\forall  x_1,x_2\in \huaA.$$
\end{defi}

Let $(A,[\cdot,\cdot,\cdot])$ and $(A^*,[\cdot,\cdot,\cdot]^*)$ be  3-Lie algebras. On $A\oplus A^*$, there is a natural nondegenerate symmetric bilinear form $(\cdot,\cdot)_+$  given by
\begin{equation} \label{eq:bf}
( x+\xi, y+\eta)_+=\langle x, \eta\rangle+\langle \xi,y\rangle,\;\;\forall  x,y\in A, \xi,\eta\in A^*.
\end{equation}
There is also a bracket operation
$[\cdot,\cdot,\cdot]_{A\oplus A^*}$ on $A\oplus A^*$ given by
\begin{eqnarray}
 \nonumber [x+\xi,y+\eta,z+\gamma]_{A\oplus A^*}&=&[x,y,z]+\ad_{x,y}^*\gamma+\ad_{y,z}^*\xi+\ad_{z,x}^*\eta\\
\label{eq:formularAA*}  &&+\add_{\xi,\eta}^*z+\add_{\eta,\gamma}^*x+\add_{\gamma,\xi}^*y+[\xi,\eta,\gamma]^*,
\end{eqnarray}
where $x,y,z\in A, \xi,\eta,\gamma\in A^*$,  $\ad^*$ and $\add^*$ are the coadjoint representations of $A$ and $A^*$ on $A^*$ and $A$ respectively. Note that the bracket operation $[\cdot,\cdot,\cdot]_{A\oplus A^*}$ is naturally invariant with respect to the symmetric bilinear form $(\cdot,\cdot)_+$, and satisfies Condition~(c) in Definition \ref{defi:Manin}. If $(A\oplus A^*,[\cdot,\cdot,\cdot]_{A\oplus A^*})$ is a 3-Lie algebra, then obviously $A$ and $A^*$ are isotropic subalgebras. Consequently, $((A\oplus A^*,(\cdot,\cdot)_+),A, A^*)$ is a Manin triple, which we call {\bf the standard Manin triple of 3-Lie algebras.}

\begin{pro}
Any Manin triple of $3$-Lie algebras is isomorphic to a standard one.
\end{pro}
\begin{proof}
For any Manin triple $((\huaA,(\cdot,\cdot)_\huaA),A_1,A_2)$,
through the nondegenerate bilinear form $(\cdot,\cdot)_\huaA$, $A_2$
is isomorphic to $A_1^*$ as a vector space.
Moreover, $A_1^*$ is equipped with a 3-Lie algebra structure from $A_2$ via this isomorphism. Then $((\huaA,(\cdot,\cdot)_\huaA),A_1,A_2)$ is
isomorphic to the standard Manin triple $((A_1\oplus
A_1^*,(\cdot,\cdot)_+),A_1,A_1^*)$.
\end{proof}

Now we turn our attention to matched pairs of 3-Lie algebras.

\begin{pro}
Let $(A,[\cdot,\cdot,\cdot])$ and $(A',[\cdot,\cdot,\cdot]')$ be
two $3$-Lie algebras. Suppose that there are skew-symmetric
linear maps $\rho:\otimes^2A\rightarrow \gl(A')$ and
$\mu:\otimes^2A'\rightarrow \gl(A)$ satisfying the following
conditions:
\begin{enumerate}
\item $(A',\rho)$ is a representation of $A$;
\item $(A,\mu)$ is a representation of $A'$;
\item For all $x_i\in A$ and $a_i\in A', 1\leq i\leq 5$, $ \rho, \mu$ satisfy the following compatibility conditions:
{\footnotesize
\begin{eqnarray}
 \mu(a_4,a_5)[x_1,x_2,x_3]&=&[\mu(a_4,a_5)x_1,x_2,x_3]+[x_1,\mu(a_4,a_5)x_2,x_3]+[x_1,x_2,\mu(a_4,a_5)x_3];\label{eq:deri1}\\
 -\mu(\rho(x_1,x_2)a_3,a_5)x_4&=&-\mu (\rho(x_1,x_4)a_5,a_3)x_2+\mu (\rho(x_2,x_4)a_5,a_3)x_1-[x_1,x_2,\mu(a_3,a_5)x_4];\label{eq:mp2}\\
 ~[\mu(a_2,a_3)x_1,x_4,x_5]&=&\mu(a_2,a_3)[x_1,x_4,x_5]+\mu(\rho(x_4,x_5)a_2,a_3)x_1+\mu(a_2,\rho(x_4,x_5)a_3)x_1;\label{eq:mp3}\\
~\rho(x_4,x_5)[a_1,a_2,a_3]'&=&[\rho(x_4,x_5)a_1,a_2,a_3]'+[a_1,\rho(x_4,x_5)a_2,a_3]'+[a_1,a_2,\rho(x_4,x_5)a_3]';\label{eq:deri2}\\
~-\rho(\mu(a_1,a_2)x_3,x_5)a_4&=&-\rho (\mu(a_1,a_4)x_5,x_3)a_2+\rho (\mu(a_2,a_4)x_5,x_3)a_1-[a_1,a_2,\rho(x_3,x_5)a_4]';\label{eq:mp5}\\
~[\rho(x_2,x_3)a_1,a_4,a_5]'&=&\rho(x_2,x_3)[a_1,a_4,a_5]'+\rho(\mu(a_4,a_5)x_2,x_3)a_1+\rho(x_2,\rho(a_4,a_5)x_3)a_1,\label{eq:mp6}
\end{eqnarray}
}
\end{enumerate}
Then there is a $3$-Lie algebra structure on $A\oplus A'$ $($as the direct sum of vector spaces$)$ defined by
\begin{eqnarray*}
[x_1+a_1,x_2+a_2,x_3+a_3]_{A\oplus A'}:&=&[x_1,x_2,x_3]+\rho(x_1,x_2)a_3+\rho(x_3,x_1)a_2+\rho(x_2,x_3)a_1\nonumber \\
&&+[a_1,a_2,a_3]'+\mu (a_1,a_2)x_3+\mu (a_3,a_1)x_2+\mu (a_2,a_3)x_1,
\label{eq:sum2}
\end{eqnarray*}
for $x_i\in A, a_i\in A', 1\leq i\leq 3$.
\end{pro}

\begin{proof}
  It follows from straightforward applications of the Fundamental Identity for the bracket operation $[\cdot,\cdot,\cdot]_{A\oplus A'}.$
\end{proof}

\begin{rmk}{\rm Eq.~(\ref{eq:deri1}) means that $\mu(a_1,a_2)$ is a derivation on the 3-Lie algebra $A$ for all $a_1,a_2\in A'$.
Eq.~(\ref{eq:deri2}) means that $\rho(x_1,x_2)$ is a derivation on the 3-Lie algebra $A'$ for all $x_1,x_2\in A$.}
\end{rmk}

\begin{defi}
 Let $(A,[\cdot,\cdot,\cdot])$ and $(A^*,[\cdot,\cdot,\cdot]^*)$ be  $3$-Lie algebras. Suppose that there are skew-symmetric linear maps $\rho:\otimes^2A\rightarrow \gl(A')$ and
$\mu:\otimes^2A'\rightarrow \gl(A)$ such that $(A',\rho)$ is a representation of $A$, $( A,\mu)$ is a representation of $A'$ and
$\rho, \mu$ satisfy Eqs.~\eqref{eq:deri1}-\eqref{eq:mp6}. Then we call $(A,A')$ (or more precisely $(A,A',\rho,\mu)$) a {\bf matched pair of $3$-Lie algebras}.
\end{defi}

We immediately get the following relation between Manin triples
and matched pairs of 3-Lie algebras:

\begin{pro} Let $(A,[\cdot,\cdot,\cdot])$ and $(A^*,[\cdot,\cdot,\cdot]^*)$ be  $3$-Lie algebras. Then $((A\oplus A^*,(\cdot,\cdot)_+),A,A^*)$ is a standard Manin triple if and only if
$(A,A^*,{\rm ad}^*,\add^*)$ is a matched pair.
\label{pp:mm}
\end{pro}

\subsection{Double construction 3-Lie bialgebras}
There are many several
 equivalent conditions for matched pairs of 3-Lie
algebras.
\begin{lem}\label{lem:3-Lie}
Let $(A,[\cdot,\cdot,\cdot])$ and $(A^*,[\cdot,\cdot,\cdot]^*)$ be  $3$-Lie algebras.  Then
$(A,A^*,{\rm ad}^*,\add^*)$ is a matched pair if and only if Eqs. \eqref{eq:deri1}, \eqref{eq:mp2} and \eqref{eq:mp3} hold for
$\rho={\rm ad}^*, \mu=\add^*$.
\end{lem}

\begin{proof} We only need to prove that for $\rho={\rm ad}^*, \mu=\add^*$, we have
$${\rm Eq.}\;\;(\ref{eq:deri1})\Longleftrightarrow {\rm Eq.}\;\;(\ref{eq:deri2}),\;\;
{\rm Eq.}\;\;(\ref{eq:mp2})\Longleftrightarrow {\rm Eq.}\;\;(\ref{eq:mp5}),\;\;
{\rm Eq.}\;\;(\ref{eq:mp3})\Longleftrightarrow {\rm Eq.}\;\;(\ref{eq:mp6}).
$$
We only prove the second equivalence since the proofs
of the other cases are similar. Let $x_1,x_2,x_4\in A$ and $a_3,a_5,a_6\in A^*$. We have
\begin{eqnarray*}
&&\langle -\add^*_{{\rm ad}^*_{x_1,x_2}a_3,a_5}x_4+\add^*_{{\rm ad}^*_{x_1,x_4}a_5,a_3}x_2-\add^*_{{\rm ad}^*_{x_2,x_4}a_5,a_3}x_1
+ [x_1,x_2,\add^*_{a_3,a_5}x_4], a_6\rangle\\
&=&\langle [{\rm ad}^*_{x_1,x_2}a_3,a_5,a_6]^*,x_4\rangle -\langle [{\rm ad}^*_{x_1,x_4}a_5,a_3,a_6]^*, x_2\rangle\\
&&+\langle [{\rm ad}^*_{x_2,x_4}a_5, a_3,a_6]^*-{\rm ad}^*_{x_2,\add^*_{a_3,a_5}x_4}a_6,x_1\rangle  \\
&=&\langle x_1,-{\rm ad}^*_{x_2,\add^*_{a_5,a_6}x_4}a_3+{\rm ad}^*_{x_4,\add^*_{a_3,a_6}x_2}a_5+[{\rm ad}^*_{x_2,x_4}a_5, a_3,a_6]^*-{\rm ad}^*_{x_2,\add^*_{a_3,a_5}x_4}a_6\rangle,
 \end{eqnarray*}
which implies the equivalence between Eq.~(\ref{eq:mp2}) and  Eq.~(\ref{eq:mp5}).
\end{proof}

\begin{thm}\label{thm:3Lieb}
Let $(A,[\cdot,\cdot,\cdot])$  be a $3$-Lie algebra and $\Delta:A\rightarrow A\otimes A\otimes A$  a linear map. Suppose that $\Delta^*:A^*\otimes A^*\otimes A^*\rightarrow A^*$
defines a $3$-Lie algebra structure $ [\cdot,\cdot,\cdot]^*$ on $A^*$.
Then $(A,A^*,{\rm ad}^*,\add^*)$ is a matched pair if and only if for all $x,y,z\in A$, the following equations are satisfied:
\begin{eqnarray}
\Delta( [x,y,z])&=&(1\otimes 1\otimes {\rm ad}_{y,z})\Delta (x)+(1\otimes 1\otimes {\rm ad}_{z,x})\Delta (y)
+(1\otimes 1\otimes {\rm ad}_{x,y})\Delta (z);\label{eq:b1}\\
\Delta ([x,y,z])&=& (1\otimes 1\otimes {\rm ad} _{y,z})\Delta (x)+(1\otimes {\rm ad}_{y,z}\otimes 1)\Delta (x)
+({\rm ad}_{y,z }\otimes 1\otimes 1)\Delta(x);\label{eq:b2}
\end{eqnarray}
\begin{equation}({\rm ad}_{x,y}\otimes 1\otimes 1+1\otimes 1\otimes {\rm ad}_{x,y})\Delta(z)=
(1\otimes {\rm ad}_{z,x}\otimes 1)\Delta(y)+(1\otimes {\rm ad}_{y,z}\otimes 1)\Delta(x)\label{eq:b3},
\end{equation}
for $x,y,z,\in A$.
\end{thm}

\begin{proof} By Lemma~\ref{lem:3-Lie}, we only need to prove that Eqs.~(\ref{eq:deri1}),    (\ref{eq:mp2}) and Eq.~(\ref{eq:mp3}) are equivalent to Eqs. \eqref{eq:b1}, \eqref{eq:b2} and \eqref{eq:b3}   respectively.

Let $\{e_1,\cdots,e_n\}$ be a basis of $A$ and $\{e_1^*,\cdots,e_n^*\}$  the dual basis. Suppose
$$[e_i,e_j,e_k]=\sum_{l=1}^nc_{ijk}^le_l, \quad [e_i^*,e_j^*,e_k^*]^*=\sum_{l=1}^nd_{ijk}^l e_l^*.$$
Then we have

$${\rm ad}^*_{e_i,e_j}e_k^*=-\sum_{l=1}^nc_{ijl}^ke_l^*;\;\;\add^*_{e_i^*,e_j^*}e_k=-\sum_{l=1}^n d_{ijl}^ke_l,\;\;
\Delta(e_k)=\sum_{i,j,l=1}^n d_{ijl}^ke_i\otimes e_j\otimes e_l.$$
By Eq.~(\ref{eq:deri1}), we have
$$\add^*_{e_\alpha^*,e_\beta^*}[e_i,e_j,e_k]-[\add^*_{e_\alpha^*,e_\beta^*}e_i,e_j,e_k]-[e_i,\add^*_{e_\alpha^*,e_\beta^*}e_j,e_k]
-[e_i,e_j,\add^*_{e_\alpha^*,e_\beta^*}e_k]=0,$$
which gives
\begin{equation}\label{eq:e1}
\sum_{l=1}^n(- d_{\alpha\beta m}^lc_{ijk}^l+d_{\alpha\beta l}^ic_{ljk}^m+d_{\alpha\beta l}^jc_{ilk}^m+d_{\alpha\beta l}^k c_{ijl}^m)=0, \quad \forall \alpha,\beta, i,j,k,m, \end{equation}
as the coefficient of $e_m$. On the other hand, the left hand side of the above equation is also
the coefficient of $e_\alpha\otimes e_\beta\otimes e_m$ in
$$(1\otimes 1\otimes {\rm ad}_{e_j,e_k})\Delta (e_i)+(1\otimes 1\otimes {\rm ad}_{e_k,e_i})\Delta (e_j)+(1\otimes 1\otimes {\rm ad}_{e_i,e_j})\Delta (e_k)-\Delta ([e_i,e_j,e_k]).$$
Thus, we deduce that Eq.~(\ref{eq:deri1}) is equivalent to Eq. \eqref{eq:b1}.

We can similarly prove that Eq.~(\ref{eq:mp2}) and Eq.~(\ref{eq:mp3}) are equivalent to Eq.~\eqref{eq:b3} and Eq.~\eqref{eq:b2} respectively.
\end{proof}

\begin{rmk}\label{rmk:eql}
{\rm In fact, Eq.~(\ref{eq:b1}) means that $\Delta$ is a 1-cocycle
associated to $(1\otimes 1\otimes {\rm ad}, A\otimes A\otimes A)$
and Eq.~(\ref{eq:b2}) means that (see Eq.~(\mref{eq:3-der}) for
the notations, also see \mcite{Bai3Liebi})
$$\Delta ([x,y,z])=[\Delta(x), y,z],\;\;\forall  x,y,z\in A.$$}
\end{rmk}

\begin{rmk}\label{re:eq} {\rm Note that there are certain symmetries in the above proof. Let $x,y,z\in A$.
\begin{enumerate}
\item[\rm(i)] Eq.~(\ref{eq:b1}) holds if and only if one of the following equations holds:
\begin{itemize}
\item $\Delta ([x,y,z])=(1\otimes {\rm ad}_{y,z}\otimes 1)\Delta (x )+( 1\otimes {\rm ad}_{z,x}\otimes 1)\Delta (y)
+(1\otimes {\rm ad}_{x,y}\otimes 1)\Delta (z)$;
\item $\Delta ([x,y,z])=({\rm ad}_{y,z}\otimes 1\otimes 1)\Delta (x)+(  {\rm ad}_{z,x}\otimes 1\otimes 1)\Delta (y)
+({\rm ad}_{x,y}\otimes 1\otimes 1)\Delta (z)$.
\end{itemize}
In fact, the former corresponds to the coefficient of $e_\alpha\otimes e_m\otimes e_\beta$ in Eq.~(\ref{eq:e1}),
whereas the latter corresponds to the coefficient of $e_m\otimes e_\alpha\otimes e_\beta$ in Eq.~(\ref{eq:e1}).
\item[\rm(ii)] Eq.~(\ref{eq:b2}) holds if and only if one of the following equations holds:
\begin{itemize}
\item $\Delta ([x,y,z])= (1\otimes 1\otimes {\rm ad}_{z,x})\Delta (y)+(1\otimes {\rm ad}_{z,x}\otimes 1)\Delta (y)
+({\rm ad}_{z,x}\otimes 1\otimes 1)\Delta(y)$;
\item $\Delta ([x,y,z])= (1\otimes 1\otimes {\rm ad} _{x,y})\Delta (z)+(1\otimes {\rm ad}_{x,y}\otimes 1)\Delta (z)
+({\rm ad}_{x,y}\otimes 1\otimes 1)\Delta(z)$.
\end{itemize}
\item[\rm(iii)] Eq.~(\ref{eq:b3}) holds if and only if one of the following equations holds:
\begin{itemize}
\item $({\rm ad}_{x,y}\otimes 1\otimes 1)\Delta(z)+(1\otimes{\rm ad}_{x,y}\otimes 1 )\Delta(z)=(1\otimes 1\otimes {\rm ad}_{z,x})\Delta(y)+(1\otimes 1\otimes {\rm ad}_{y,z})\Delta(x);$
\item $(1\otimes {\rm ad}_{x,y}\otimes 1)\Delta(z)+(1\otimes 1\otimes {\rm ad}_{x,y})\Delta(z)=({\rm ad}_{z,x}\otimes 1\otimes 1)\Delta(y)+( {\rm ad}_{y,z}\otimes 1 \otimes 1)\Delta(x).$
\end{itemize}
\end{enumerate}
}
\end{rmk}

The three equations \eqref{eq:b1} -- \eqref{eq:b3} are not independent.
\begin{pro}\label{pp:equiv}
Any two equations in Eqs.~  \eqref{eq:b1}, \eqref{eq:b2} and \eqref{eq:b3} imply the third one.
\end{pro}
\begin{proof}  Substituting Eq.~\eqref{eq:b2} into Eq.~\eqref{eq:b1} yields the first equation in Item (iii) in the above remark, which is equivalent to Eq. (\ref{eq:b3}).

Substituting  Eq.~(\ref{eq:b3}) into Eq. (\ref{eq:b2}) yields the first equation in Item (i) in the above remark, which is equivalent to Eq. (\ref{eq:b1}).

Eq.~(\ref{eq:b1}) and the first equation in Item (iii) in the above remark imply Eq.~(\ref{eq:b2}). Hence by the same remark, Eqs.~(\ref{eq:b1}) and (\ref{eq:b3})  imply Eq.~(\ref{eq:b2}).
\end{proof}

Since a Manin triple of 3-Lie algebras $((A\oplus A^*,(\cdot,\cdot)_+),A,A^*)$ gives a natural pseudo-metric 3-Lie algebra structure on the double space $A\oplus A^*$, we give the following definition:

\begin{defi}
{\rm
  Let $A$ be a $3$-Lie algebra and $\Delta:A\rightarrow A\otimes A\otimes A$  a linear map. Suppose that $\Delta^*: A^*\otimes A^*\otimes A^*\rightarrow A$
defines a $3$-Lie algebra structure on $A^*$. If $\Delta$ satisfies Eqs.~\eqref{eq:b1} and \eqref{eq:b2}, then we call $(A,\Delta)$ a {\bf \typeII}.
}
\end{defi}

Combining Proposition~\ref{pp:mm}, Theorem~\ref{thm:3Lieb} and Proposition~\ref{pp:equiv}, we obtain

\begin{thm}\label{thm:relations}
Let $A$ be a $3$-Lie algebra and $\Delta:A\rightarrow A\otimes A\otimes A$  a linear map. Suppose that $\Delta^*:A^*\otimes A^*\otimes A^*\rightarrow A^*$
defines a $3$-Lie algebra structure on $A^*$. Then the following statements are equivalent.
\begin{enumerate}
\item $(A,\Delta)$ is a \typeII;
\item $((A\oplus A^*,(\cdot,\cdot)_+),A,A^*)$ is a standard Manin triple, where the bilinear form $(\cdot,\cdot)_+$ and the $3$-Lie bracket $[\cdot,\cdot,\cdot]_{A\oplus A^*}$ are given by Eqs.~\eqref{eq:bf} and \eqref{eq:formularAA*} respectively;
\item $(A,A^*,{\rm ad}^*,\add^*)$ is a matched pair of $3$-Lie algebras.
\end{enumerate}
\end{thm}

\subsection{The relationship between the two types of $3$-Lie bialgebras and examples}
We begin with a relationship between the \typeI and the \typeII.
We have the following result slightly improving Remark~\ref{re:eq}.(i):

\begin{pro}Let $A$ be a $3$-Lie algebra and $\Delta:A\rightarrow A\otimes A\otimes A$  a linear map. Suppose that $\Delta^*:A^*\otimes A^*\otimes A^*\rightarrow A^*$
defines a skew-symmetric operation on $A^*$. Then the following
conditions are equivalent:
\begin{enumerate}
\item $\Delta$ is a $1$-cocycle associated to $(A\otimes A\otimes A, {\rm ad}\otimes
1\otimes 1)$;

\item $\Delta$ is a $1$-cocycle
associated to $( A\otimes A\otimes A,1\otimes {\rm ad}\otimes 1)$;

\item$\Delta$ is a $1$-cocycle associated to $(A\otimes A\otimes A, 1\otimes 1\otimes
{\rm ad})$.
\end{enumerate}
\end{pro}

\begin{proof}
We only prove that (a) holds if and only if (b) holds. The proof
of the other cases is similar. By the proof in
Theorem~\ref{thm:3Lieb}, we have
\begin{eqnarray*}
(a)&\Leftrightarrow&\sum_{\alpha,\beta, m}\sum_{l}[-
d_{\alpha\beta m}^lc_{ijk}^l+d_{\alpha\beta
l}^ic_{jkl}^m+d_{\alpha\beta l}^jc_{kil}^m+d_{\alpha\beta l}^k
c_{ijl}^m] e_\alpha\otimes
e_\beta\otimes e_m=0;\\
(b)&\Leftrightarrow &\sum_{\alpha,\beta, m}\sum_{l}[-
d_{\alpha m \beta}^lc_{ijk}^l+d_{\alpha l \beta
}^ic_{jkl}^m+d_{\alpha l \beta }^jc_{kil}^m+d_{\alpha l \beta }^k
c_{ijl}^m] e_\alpha\otimes
e_m\otimes e_\beta=0.
\end{eqnarray*}
The skew-symmetry of $\Delta^*$ means $d_{\alpha\beta m}=-d_{\alpha m\beta}$, which implies that (a) holds if and
only if (b) holds. \end{proof}

As a direct consequence, we obtain

\begin{cor}
A \typeII gives a \typeI.
\end{cor}

\begin{proof}
Let $(A,\Delta)$ be a \typeII. Let $k_1, k_2, k_3$ be
 complex numbers such that $k_1+k_2+k_3=1$. Denote $\Delta_i=k_i\Delta, i=1, 2, 3$. Then by definition
$\Delta=\Delta_1+\Delta_2+\Delta_3$ defines a \typeI.
\end{proof}

\begin{rmk} {\rm From the corollary, it is natural to consider whether there
is a ``coboundary double construction 3-Lie bialgebras". As in
Section 2, for any $r=\sum_ix_i\otimes y_i\in A\otimes A$, set
$$\Delta_1(x)=\sum_{i,j} [x,x_i,x_j]\otimes y_j\otimes
y_i, \Delta_2(x)=\sum_{i,j} y_i\otimes [x,x_i,x_j]\otimes y_j,
\Delta_3(x)=\sum_{i,j} y_j\otimes y_i\otimes [x,x_i,x_j], $$ for
any $x\in A$. If all  $\Delta_i^*:A^*\otimes A^*\otimes A^*\longrightarrow A^*$ are skew-symmetric, it is straightforward to see that
$\Delta_1=\Delta_2=\Delta_3$. Unfortunately, there is no natural condition for $r$ such that $\Delta^*$ defines a
skew-symmetric operation on $A^*$ (in fact, it involves the
concrete structural constants of $A$). Furthermore, if such an $r$ exists, then we have a ``modified" version of
Theorem \ref{thm:Condition on r}. However, an algebraic equation for $r$ remains to be discovered.}
\label{mk:cbdc}
\end{rmk}

We end the paper with some examples to illustrate these points and various other phenomena of \typeII. As a first example, we have

\begin{ex}{\rm
For any 3-Lie algebra $A$, taking $\Delta=0$, then $(A,\Delta)$ is a
\typeII. In this case, the corresponding Manin triple gives a
pseudo-metric 3-Lie algebra $(A\ltimes_{{\rm ad}^*}
A^*,(\cdot,\cdot)_+)$. Dually, for any trivial 3-Lie algebra $A$
(namely whose product is zero), any 3-Lie algebra structure
$\Delta^*$ on the dual space $A^*$ makes $(A,\Delta)$  a \typeII. Such \typeIIs are called {\bf trivial \typeIIs}. }
\end{ex}

In general, it is not easy to determine whether there exists a non-trivial \typeII on a given non-trivial 3-Lie algebra. One reason is that there is certain
inconsistence between the 1-cocycle or Eq.~(\ref{eq:b2}) and
skew-symmetry. Before giving such examples, we recall the following results on the classification of complex 3-Lie algebras of dimensions 3 and 4.

\begin{pro} \mcite{BSZ}
\begin{enumerate}
\item
There is a unique non-trivial 3-dimensional complex 3-Lie algebra. It has a basis $\{e_1,e_2,e_3\}$ with respect to which the non-zero product is given by
$[e_1,e_2,e_3]=e_1.$
\item
Let $A$ be a non-trivial 4-dimensional complex 3-Lie algebra. Then $A$ has a basis $\{e_1,e_2,e_3,e_4\}$ with respect to which the product of the 3-Lie algebra is given by one of the following.
\begin{itemize}
\item[\rm(1)] $[e_1,e_2,e_3]=e_4, \quad [e_2,e_3,e_4]=e_1,\quad
[e_1,e_3,e_4]=e_2,\quad[e_1,e_2,e_4]=e_3$.

\item[(2)] $[e_1,e_2,e_3]=e_1$.

\item[(3)] $[e_2,e_3,e_4]=e_1$.

\item[(4)] $[e_2,e_3,e_4]=e_1, \quad[e_1,e_3,e_4]=e_2$.

\item[(5)] $[e_2,e_3,e_4]=e_2, \quad[e_1,e_3,e_4]=e_1$.

\item[(6)] $[e_2,e_3,e_4]=\alpha e_1+e_2,\;\alpha\ne
0,\quad[e_1,e_3,e_4]=e_2$.

\item[(7)] $[e_2,e_3,e_4]=e_1,\quad [e_1,e_3,e_4]=e_2,\quad[e_1,e_2,e_4]=e_3$.
\end{itemize}
\end{enumerate}
\label{pp:4}
\end{pro}

In the following two examples, the \typeIIs have only the zero coproduct, exhibiting the aforementioned inconsistence between
1-cocycle and skew-symmetry. Compare them with Example~\mref{ex:3d}
which gives the \typeIs with non-zero coproduct on the non-trivial
3-dimensional complex 3-Lie algebra.

\begin{ex}{\rm For the unique non-trivial 3-dimensional complex 3-Lie algebra in Proposition~\ref{pp:4}, a linear map
$\Delta: A\rightarrow A\otimes A\otimes A$ satisfies
Eq.~(\ref{eq:b1}) such that $\Delta^*:A^*\otimes A^*\otimes
A^*\rightarrow A^*$ defines a skew-symmetric operation on $A^*$,
then $\Delta=0$.
}
\end{ex}

\begin{ex}{\rm
For classes (2), (5) and (6) in Proposition~\ref{pp:4}, if a linear map $\Delta: A\rightarrow A\otimes A\otimes A$ satisfies
Eq.~(\ref{eq:b1}) and Eq.~(\ref{eq:b2}) such
that $\Delta^*:A^*\otimes A^*\otimes A^*\rightarrow A^*$ defines a
skew-symmetric operation on $A^*$, then $\Delta=0$.}
\end{ex}

The following is an example of non-trivial \typeIIs and an illustration of Remark~\ref{mk:cbdc}.

\begin{ex} {\rm
Consider the first class of the 4-dimensional non-trivial 3-Lie algebras in Proposition~\ref{pp:4}.  Note that $A$ is the (unique) simple complex 3-Lie algebra whose only ideals are the zero ideal
and $A$ itself. Define a linear map $\Delta: A\rightarrow A\otimes
A\otimes A$ by (see Example~\mref{ex:3d} for notations)
$$\Delta(e_1)=e_2\wedge e_3\wedge e_4,\;
\Delta(e_2)=e_1\wedge e_3\wedge e_4,\; \Delta(e_3)=e_1\wedge
e_2\wedge e_4,\;\Delta(e_4)=e_1\wedge e_2\wedge e_3.
$$ Then
$\Delta$ satisfies Eqs.~\eqref{eq:b1} and \eqref{eq:b2}. Moreover,
$\Delta^*:A^*\otimes A^*\otimes A^*\rightarrow A^*$ defines a
$3$-Lie algebra structure on $A^*$ which is isomorphic to $A$. So
$(A,\Delta)$ is a \typeII.

Furthermore, let
$$r=e_1\otimes e_1-e_2\otimes e_2+e_3\otimes
e_3-e_4\otimes e_4=\sum_i(-1)^{i+1}e_i\otimes e_i.$$ Note that $r$
is symmetric. We have
$$\Delta(e_i)=\sum_{k,l} (-1)^{k+l} [e_i,e_k,e_l]\otimes e_l\otimes e_k=\Delta_1(e_i), \;\;1\leq i\leq 4,$$ where $\Delta_1$ is given by Eq.~\eqref{eq:delta123}. However, for the eight items in Theorem \ref{thm:Condition on r}, every item is
not zero, but the sum is zero.
}
\end{ex}

Moreover, this \typeII structure can be induced by an $r\in A\otimes A$. But unfortunately,
such an $r$ does not satisfy the 3-Lie CYBE or even a well-constructed algebraic equation on $A$.

The other classes in Proposition~\ref{pp:4}
also give non-trivial \typeIIs.
All these examples follow from straightforward
computations. For details see~\mcite{DB} which is based on the general framework in this paper.

\smallskip

\noindent
{\bf Acknowledgements.}
This research  is supported by the Natural Science Foundation of China (11471139, 11271202, 11221091, 11425104,  11371178), SRFDP
(20120031110022) and the Natural Science Foundation of Jilin Province (20140520054JH).

\end{document}